\pdfoutput=1
\PassOptionsToPackage{pdftex,
pdfencoding=auto,
pdfnewwindow=true,
pdfusetitle=true,
bookmarks=true,
bookmarksnumbered=true,
bookmarksopen=false,
pdfpagemode=UseThumbs,
bookmarksopenlevel=1,
pdfpagelabels=true,
breaklinks=true,
colorlinks=true,
unicode=true,
pdfborder={0 0 0},
backref=false
}{hyperref}
\PassOptionsToPackage{dvipsnames,pdftex}{xcolor}
\documentclass[superscriptaddress,aps,pra,nofootinbib,twocolumn,twoside,reprint,letterpaper,longbibliography,showpacs,10pt,floatfix,tikz,amsmath,amsfonts,amssymb]{revtex4-2}

\AtBeginDocument{
    \newwrite\bibnotes
    \def\bibnotesext{Notes.bib}
    \immediate\openout\bibnotes=\jobname\bibnotesext
    \immediate\write\bibnotes{@CONTROL{REVTEX42Control}}
    \immediate\write\bibnotes{@CONTROL{
    apsrev42Control,author="08",editor="1",pages="1",title="0",year="1"}}
     \if@filesw
     \immediate\write\@auxout{\string\citation{apsrev42Control}}
    \fi
}

\usepackage{amsthm,mathtools,bm,graphicx,verbatim,mathrsfs,units}
\usepackage[utf8]{inputenx}
\usepackage{bbm}
\usepackage{xcolor}
\usepackage{dsfont}
\usepackage{float}
\usepackage{cancel}
\usepackage{booktabs}
\usepackage{mathrsfs}
\usepackage[alwaysadjust,flushleft,pointlessenum]{paralist}
\setdefaultenum{\bfseries1.}{}{}{}
\usepackage{changepage}
\usepackage[nolist,withpage]{acronym}

\definecolor{myurlcolor}{rgb}{0,0,0.7}
\definecolor{myrefcolor}{rgb}{0.8,0,0}
\definecolor{purple}{RGB}{128,0,128}
\definecolor{ultramarine}{RGB}{63, 0, 255}
\definecolor{medblue}{RGB}{0, 0, 100}
\definecolor{googleblue}{RGB}{34, 0, 204}
\definecolor{panblue}{RGB}{0,24,150}
\definecolor{carmine}{RGB}{150, 0, 24}
\definecolor{gray}{RGB}{150, 150, 150}
\usepackage{hyperref} 
\hypersetup{
unicode=true,
bookmarksnumbered=false,
bookmarksopen=false,
breaklinks=false,
colorlinks=true,
linkcolor=carmine,
citecolor=googleblue,
urlcolor=panblue,
anchorcolor=OliveGreen}
\usepackage{amsthm}
\usepackage{amsmath}
\usepackage{amssymb}
\usepackage{refcount}
\usepackage{graphicx}
\usepackage{color}
\usepackage{enumitem} 
\usepackage[caption=false]{subfig}
\usepackage{tikz}
\usetikzlibrary{shapes.geometric,shapes,positioning,shapes.misc}

\usepackage{pifont}
\newcommand{\cmark}{\ding{51}}%
\newcommand{\xmark}{\ding{55}}%
\usepackage{tabularx}
\usepackage{array}
\newtheorem*{principle*}{Principle}

\newcommand{\W}{\mathrm{\W}}
\newcommand{\GHZ}{\mathrm{\GHZ}}

\makeatletter
\DeclareRobustCommand{\Tr}{\operatorname{\mathrm{Tr}}\@ifstar\@firstofone\@Pr}
\newcommand{\@Tr}[1]{\ensuremath{\left(#1\right)}}
\makeatother

\newcommand{\beq}{\begin{equation}}
\newcommand{\eeq}{\end{equation}}
\newcommand{\beqa}{\begin{eqnarray}}
\newcommand{\eeqa}{\end{eqnarray}}
\newcommand{\bra}[1]{\ensuremath{\left\langle#1\right|}}
\newcommand{\ket}[1]{\ensuremath{\left|#1\right\rangle}}

\renewcommand{\today}{\number\day\space\ifcase\month\or
   January\or February\or March\or April\or May\or June\or
   July\or August\or September\or October\or November\or December\fi
   \space\number\year}

\definecolor{SkyBlue}{RGB}{135,206,235}

\usepackage{cprotect} %

\newtheorem{thm}{Theorem}
\newtheorem{theorem}[thm]{Theorem}

\newtheorem{proposition}[thm]{Proposition}
\newtheorem{lemma}[thm]{Lemma}

\newtheorem{definition}[thm]{Definition}

\newtheorem{theoreminner}{Theorem}
\makeatletter
\newenvironment{theoremrevisited}[1]{%
  \renewcommand{\thetheoreminner}{\getrefnumber{#1}}%
  \renewcommand{\theHtheoreminner}{technical.\getrefnumber{#1}}%
  \begin{theoreminner}[technical version]%
  \def\@currentlabelname{\getrefnumber{#1} (technical version)}%
}{%
  \end{theoreminner}%
}
\makeatother

\makeatletter
\newenvironment{theoremnoise}[1]{%
  \renewcommand{\thetheoreminner}{\getrefnumber{#1}}%
  \renewcommand{\theHtheoreminner}{noise.\getrefnumber{#1}}%
  \begin{theoreminner}[noise-robust version]%
  \def\@currentlabelname{\getrefnumber{#1} (noise-robust version)}%
}{%
  \end{theoreminner}%
}
\makeatother

\makeatletter
\newenvironment{replemma}[1]{%
  \begingroup
  \def\thethm{\ref*{#1}}%
  \begin{lemma}
}{%
  \end{lemma}
  \endgroup
}
\makeatother

\newtheoremstyle{defblock}{0.7\topsep}{0pt}{}{}{}{: }{0pt plus 1pt minus 1pt}{\thmname{\bfseries{#1}}\thmnumber{\bfseries{#2}}\color{medblue}\bfseries\thmnote{#3}}
\theoremstyle{defblock}

\theoremstyle{remark}

\newcommand{\tagprop}[1]{\tag{\hyperref[#1]{P\ref{#1}}}}

\usepackage{microtype}
\microtypecontext{spacing=nonfrench}
\microtypesetup{
expansion={true,nocompatibility},
protrusion={true,nocompatibility},
activate={true,nocompatibility},
tracking=true,
kerning=true,
spacing={true}
}
\usepackage[all=normal,floats=tight,mathspacing=tight,wordspacing=tight,paragraphs=normal,tracking=tight,charwidths=tight,mathdisplays=normal,sections=normal,margins=normal]{savetrees}
\everypar=\expandafter{\the\everypar\loosness=-1 }
\usepackage[style=american,autopunct=true]{csquotes}

\usepackage[scr=boondoxupr]{mathalfa} %

\usepackage{hyperref}
\hypersetup{
linkcolor=carmine,
citecolor=googleblue,
urlcolor=panblue,
anchorcolor=OliveGreen}

\AtBeginDocument{%
    \newwrite\bibnotes
    \def\bibnotesext{Notes.bib}
    \immediate\openout\bibnotes=\jobname\bibnotesext
    \immediate\write\bibnotes{@CONTROL{REVTEX42Control}}
    \immediate\write\bibnotes{@CONTROL{%
    apsrev42Control,editor="0",pages="0",title="0",year="1"}}
     \if@filesw
     \immediate\write\@auxout{\string\citation{apsrev42Control}}%
    \fi
}

\usepackage{placeins}%

\begin{document}

\begin{abstract}

We introduce the coordination principle, which states that perfect coordination, in the form of agreement on a uniformly random output, among N parties is possible only if they share a common cause. This principle is purely causal and can be viewed as a multipartite generalization of Reichenbach’s 
common cause principle. We prove that quantum information theory satisfy the
coordination principle in any network, and derive noise-tolerant Bell-like inequalities 
that certify the presence of a common cause. We further show that the principle is not a consequence of 
no-signaling and independence alone by constructing a concrete operational probabilistic theory that 
obeys both principles while still allowing perfect coordination without a common cause. 
This possibility arises only in fully general causal scenarios with intermediate transformations between 
preparations and measurements. We also formulate a genuinely quantum coordination task, showing that 
the preparation of a multipartite GHZ state requires a quantum common cause, which can be certified by 
Bell-like inequalities which are experimentally testable. Finally, we discuss the open problem of finding a quantitative, noise-tolerant version of the coordination principle that constrains approximate coordination in any reasonable causal theory. This work is the extended version of the more compact letter Ref.~\cite{CompanionPaperPRL} and provides all the technical details of the proofs.

\end{abstract}

\title{A missing causal principle: Coordination}

\author{Daniel Centeno$^{\dagger,*}$}
\affiliation{Perimeter Institute for Theoretical Physics, Waterloo, Ontario, Canada.}
\affiliation{Department of Physics and Astronomy, University of Waterloo, Waterloo, Ontario, Canada, N2L 3G1}

\author{Antoine Coquet$^{\dagger,\star}$}
\affiliation{Inria, CPHT, LIX, CNRS, École polytechnique, Institut Polytechnique de Paris, Palaiseau, France}

\author{Maria~Ciudad~Alañón}
\affiliation{Perimeter Institute for Theoretical Physics, Waterloo, Ontario, Canada.}
\affiliation{Department of Physics and Astronomy, University of Waterloo, Waterloo, Ontario, Canada, N2L 3G1}

\author{Lucas Tendick}
\affiliation{Inria, CPHT, LIX, CNRS, École polytechnique, Institut Polytechnique de Paris, Palaiseau, France}

\author{Marc-Olivier Renou}
\affiliation{Inria, CPHT, LIX, CNRS, École polytechnique, Institut Polytechnique de Paris, Palaiseau, France}

\author{Elie Wolfe}
\affiliation{Perimeter Institute for Theoretical Physics, Waterloo, Ontario, Canada.}
\affiliation{Department of Physics and Astronomy, University of Waterloo, Waterloo, Ontario, Canada, N2L 3G1}

\begingroup
\renewcommand{\thefootnote}{$\dagger$}
\footnotetext[0]{These authors contributed specially significantly.}
\endgroup

\begingroup
\renewcommand{\thefootnote}{$*$}
\footnotetext[0]{dcentenodiaz@perimeterinstitute.ca}
\endgroup

\begingroup
\renewcommand{\thefootnote}{$\star$}
\footnotetext[0]{antoine.coquet@inria.fr}
\endgroup

\maketitle

\section{Introduction}
Studying correlations between different phenomena in nature lies at the heart of many scientific disciplines, and in particular of physics. A core aspect of this endeavor is to identify causal explanations for these correlations~\cite{Pearl_2009}. Bell’s seminal work~\cite{bell1964einstein} showed that classical causal models cannot reproduce all correlations predicted by quantum theory, as subsequently confirmed in experiments \cite{Shalm2015LoopholeFree,Giustina2015LoopholeFree,Hensen2015LoopholeFree,Rosenfeld2017LoopholeFree}. This insight has led to substantial progress, both at the foundational level \cite{Brunner2014BellReview, tavakoli2022bell, Masanes2006NoSignaling} and in the emergence of new applications~\cite{Supic2020selftestingof,Ekert1991Crypto,Acin2007DIQKD, Colbeck2009Randomness, Pironio2010Randomness}.

Although classical and quantum theories yield different predictions, both satisfy the principles of No-Signaling and Independence (NSI)~\cite{Coiteux2021any, Coiteux2021no, Henson2014, Wolfe2016inflation, Gisin2020NSI, Reichenbach,popescu1994quantum}. The No-Signaling principle expresses the impossibility of faster-than-light communication: operationally, a variable cannot be influenced by events outside its causal past. The Independence principle states that two variables may be correlated only when they have a common cause\footnote{Note that in this work we use the term ``cause'' to refer to direct and indirect causes. More formally, 
$X$ is a cause of a variable 
$Y$ if there exists a directed causal path from $X$ to $Y$ in the circuit under consideration. For example, in Fig.~\ref{2-layer-tetrahedron}, the \emph{causes} of $A_1$ are the transformations $\neg(A_2A_3)$, $\neg(A_2A_4)$, $\neg(A_3A_4)$ and the sources $\neg(A_2)$, $\neg(A_3)$ and $\neg(A_4)$.} in their causal past. The NSI principles are viewed as minimal constraints that any reasonable theory of information should satisfy. Most hypothetical post-quantum theories of information obey NSI, such as boxworld~\cite{Barrett2007BoxWorld,Janotta2012BoxWorld}. However, no correlations beyond those predicted by quantum information theory have been observed to date, which motivates the search for additional principles beyond NSI that further constrain the correlations realized in nature.

Several attempts have been made to identify principles that explain why post-quantum correlations are not found in experiments~\cite{POPESCU1992Generic, Pawowski2009Information, Fritz2013Local, Brassard2006Limit, Linden2007Computation,Gisin2020NSI}. These proposals take different forms, ranging from complexity-theoretic constraints to information-theoretic ones. Here, however, we do not aim to reconstruct the set of quantum correlations from such assumptions. Rather, since causality-based principles are arguably more fundamental and rely on fewer auxiliary assumptions, we ask whether there exists an additional causal principle, beyond the NSI principles, that must be satisfied by any reasonable theory of information.

A natural starting point is the contrapositive of the Independence principle, namely Reichenbach’s common cause principle~\cite{Reichenbach}, which states that if two variables are \emph{correlated}, then they must share a common cause. However, extending this idea to the multipartite setting is not straightforward. Even when $N$ variables are jointly correlated (i.e., their distribution does not factorize), Reichenbach-style reasoning enforces only pairwise common causes and does not, by itself, imply the existence of a single common cause shared by all $N$ variables. This limitation has been pointed out repeatedly in the literature~\cite{henson2005comparing,uffink1999principle,steudel2015information,Henson2014}, and motivates the search for genuinely multipartite causal principles.

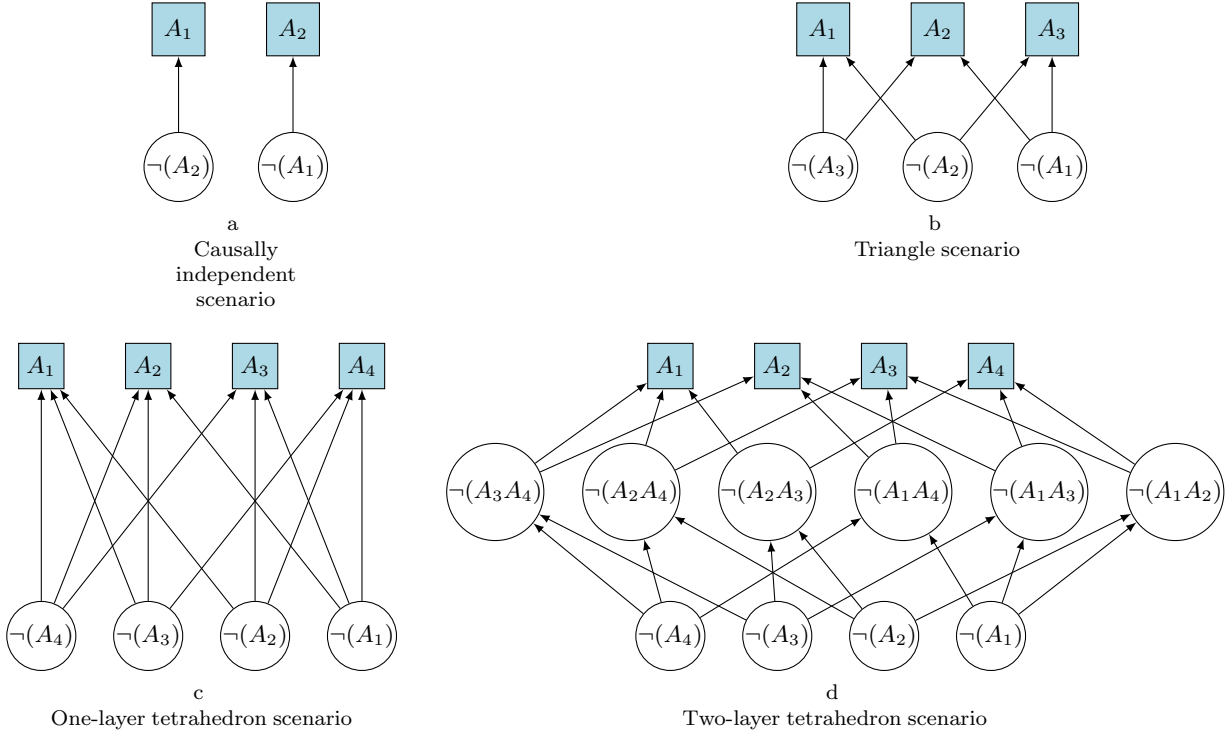
\begin{figure*}[t]
\rule{\textwidth}{0.1pt}
     \centering
     \subfloat[label1][\linebreak Causally independent scenario\label{no_correlation}]{\begin{tikzpicture}
     \definecolor{lightblue}{RGB}{173,216,230}
     \definecolor{lightpink}{RGB}{255,182,193}
  \node[draw, rectangle, fill=lightblue, minimum size = 7mm] (A) {$A_1$};
  \node[draw, rectangle, fill=lightblue, minimum size = 7mm, right=0.8cm of A] (B) {$A_2$};
  \node[draw, circle, inner sep=0pt, below=1cm of A] (A') {$\neg(A_2)$};
  \node[draw, circle, inner sep=0pt,  below=1cm of B] (B') {$\neg(A_1)$};
  \draw[->, >=latex] (A') -- (A);
  \draw[->, >=latex] (B') -- (B);
     \end{tikzpicture}}
      \hspace{6cm}
     \subfloat[triangle][\linebreak Triangle scenario\label{triangle}]{\begin{tikzpicture}
         \definecolor{lightblue}{RGB}{173,216,230}
     \definecolor{lightpink}{RGB}{255,182,193}
  \node[draw, rectangle, fill=lightblue, minimum size = 7mm] (A) {$A_1$};
  \node[draw, rectangle, fill=lightblue, minimum size = 7mm, right=0.8cm of A] (B) {$A_2$};
  \node[draw, rectangle, fill=lightblue, minimum size = 7mm, right=0.8cm of B] (C) {$A_3$};
  \node[draw, circle, inner sep=0pt, below=1cm of A] (AB) {$\neg(A_3)$};
  \node[draw, circle, inner sep=0pt,  below=1cm of B] (AC) {$\neg(A_2)$};
  \node[draw, circle, inner sep=0pt,  below=1cm of C] (BC) {$\neg(A_1)$};
  \draw[->, >=latex] (AB) -- (A);
  \draw[->, >=latex] (AB) -- (B);
  \draw[->, >=latex] (AC) -- (A);
  \draw[->, >=latex] (AC) -- (C);
  \draw[->, >=latex] (BC) -- (C);
  \draw[->, >=latex] (BC) -- (B);
     \end{tikzpicture}}\\
     \subfloat [tetrahedron][\linebreak One-layer tetrahedron scenario\label{1-layer-tetrahedron}]{\begin{tikzpicture}[
  node distance=1cm and 1.5cm,
  every node/.style={draw, circle, minimum size=6mm, inner sep=0pt},
  level 1/.style={sibling distance=15mm},
  level 2/.style={sibling distance=10mm},
  arrow/.style={<-, >=latex}
]
  \definecolor{lightblue}{RGB}{173,216,230}
  \definecolor{lightpink}{RGB}{255,182,193}

  \node[draw, rectangle, fill=lightblue] (A) {$A_1$};
  \node[draw, rectangle, fill=lightblue, right=0.8cm of A] (B) {$A_2$};
  \node[draw, rectangle, fill=lightblue, right= 0.8cm of B] (C) {$A_3$};
  \node[draw, rectangle, fill=lightblue, right=0.8cm of C] (D) {$A_4$};

  \node[below=2.8cm of A] (ABC) {$\neg (A_4)$};
  \node[below=2.8cm of B] (ABD) {$\neg (A_3)$};
  \node[below=2.8cm of C] (ACD) {$\neg (A_2)$};
  \node[below=2.8cm of D] (BCD) {$\neg (A_1)$};

  \foreach \s/\d in {A/ABC, A/ACD, A/ABD, B/ABC, B/BCD, B/ABD, C/ABC, C/BCD, C/ACD, D/ABD, D/BCD, D/ACD} {\draw[arrow] (\s) -- (\d);}

\end{tikzpicture}}
     \qquad
     \subfloat [full-tetrahedron][\linebreak Two-layer tetrahedron scenario\label{2-layer-tetrahedron}]{\begin{tikzpicture}[
  node distance=1cm and 1.5cm,
  every node/.style={draw, circle, minimum size=6mm, inner sep=0pt},
  level 1/.style={sibling distance=15mm},
  level 2/.style={sibling distance=10mm},
  arrow/.style={<-, >=latex}
]
  \definecolor{lightblue}{RGB}{173,216,230}
  \definecolor{lightpink}{RGB}{255,182,193}

  \node[draw, rectangle, fill=lightblue] (A) {$A_1$};
  \node[draw, rectangle, fill=lightblue, right=0.8cm of A] (B) {$A_2$};
  \node[draw, rectangle, fill=lightblue, right= 0.8cm of B] (C) {$A_3$};
  \node[draw, rectangle, fill=lightblue, right=0.8cm of C] (D) {$A_4$};

  \node[below left=0.9cm and 1.55cm of A] (AB) {$\neg (A_3A_4)$};
  \node[right=0.5cm of AB] (AC) {$\neg (A_2A_4)$};
  \node[right=0.5cm of AC] (AD) {$\neg (A_2A_3)$};
  \node[right=0.5cm of AD] (BC) {$\neg (A_1A_4)$};
  \node[right=0.5cm of BC] (BD) {$\neg (A_1A_3)$};
  \node[right=0.5cm of BD] (CD) {$\neg (A_1A_2)$};

  \node[below=2.8cm of A] (ABC) {$\neg (A_4)$};
  \node[below=2.8cm of B] (ABD) {$\neg (A_3)$};
  \node[below=2.8cm of C] (ACD) {$\neg (A_2)$};
  \node[below=2.8cm of D] (BCD) {$\neg (A_1)$};

  \foreach \s/\d in {A/AB, A/AC, A/AD, B/AB, B/BC, B/BD, C/AC, C/BC, C/CD, D/AD, D/BD, D/CD} {\draw[arrow] (\s) -- (\d);}

  \foreach \s/\d in {AB/ABC, AB/ABD, AC/ABC, AC/ACD, AD/ABD, AD/ACD, BC/ABC, BC/BCD, BD/ABD, BD/BCD, CD/ACD, CD/BCD} {\draw[arrow] (\s) -- (\d);}

\end{tikzpicture}}\\

        \caption{Various causal structures represented as directed acyclic graphs (DAGs). Squares denote observed nodes and circles denote latent nodes. Latent nodes represent either sources of information (when parentless) or transformations of information (when they have parents), while observed nodes correspond to measurements. The label notation of the nodes is chosen to make causal relationships explicit. Concretely, we label latent nodes by the complement of the set of observed nodes in their future, denoted $\neg(\cdot)$, whereas the observed nodes are simply denoted by $A_1,...,A_N$.  Reichenbach’s common cause principle rules out perfect coordination (indeed, any non-factorizing distribution) in (a), where $A_1$ and $A_2$ are causally independent and share no common cause. In contrast, in more complex networks without intermediate transformations, such as (b) the triangle and (c) the one-layer tetrahedron, this simple reasoning does not directly exclude perfect coordination; instead, Refs.~\cite{Henson2014,Coiteux2021any} showed that perfect coordination is impossible there under No-Signaling and Independence together with device replication. Crucially, these results do not cover scenarios with intermediate transformations. Here we show that NSI (plus device replication) do not exclude perfect coordination in (d), the two-layer tetrahedron, and we also show that, when one imposes quantum theory, perfect coordination is impossible even in the presence of intermediate transformations. 
        Table~\ref{Table d} summarizes the present paragraph.}
        \label{different-scenarios}
        \rule{\textwidth}{0.1pt}
\end{figure*}

In this work, we introduce a causality principle based on the notion of \emph{coordination} among $N$ parties, that is, the ability of several parties to generate and agree on the same randomly chosen output. We call it the \emph{Coordination Principle}.

\begin{principle*}[Coordination]
$N$ perfectly coordinated variables must share a common cause.
\end{principle*}

This principle has been proved to be respected in classical theories~\cite{steudel2015information}\footnote{This principle was introduced under the name of the \emph{extended common cause principle} for classical theories in Ref.~\cite{steudel2015information}.}. A natural next step is to ask whether it also holds in quantum theory and, more generally, whether it can be derived from the standard causal principles of No-Signaling and Independence.

Previous works concluded that for \emph{some} causal structures the Coordination Principle indeed follows from NSI when one assumes the possibility of device replication. Concretely, Ref.~\cite{Henson2014} showed that in the triangle scenario of Fig.~\ref{triangle}, perfect coordination among the three players is impossible even if the sources distribute post-quantum resources (for instance, PR boxes~\cite{popescu1994quantum}). Ref.~\cite{Coiteux2021any} generalized this result to causal structures with an arbitrary number of parties that share no common cause, including the four party case illustrated in Fig.~\ref{1-layer-tetrahedron}. However, these results were restricted to causal structures without \emph{intermediate transformations}. Such transformations have since been recognized as essential for a fully general description of nonclassical causal models~\cite{centeno2024significance}. In this paper, we revisit the problem while explicitly allowing intermediate transformations. 

Our first result concerns quantum theory: using the quantum inflation technique~\cite{wolfe2021quantum}, we prove that the Coordination Principle still holds in these fully general causal structures. In other words, we show that in the most general causal scenario with an arbitrary number of parties that do not share a common cause, perfect coordination is impossible within quantum theory. Notably, our proof yields a simple, experimentally friendly Bell-type inequality that provides a noise-robust certificate of our
findings.

Then, our second result shows that, once intermediate transformations are admitted, the Coordination Principle is \emph{not} a direct consequence of NSI. We establish this fact by constructing an explicit operational probabilistic theory (OPT) in which No-Signaling and Independence hold, yet perfect coordination is possible without a common cause, thereby violating the Coordination Principle. More precisely, our construction allows for perfect coordination on the causal structure of Fig.~\ref{2-layer-tetrahedron}, which serves as an explicit counterexample since there is no common cause shared by the four parties.

We then turn to a genuinely quantum analogue of perfect coordination, and ask whether a multipartite GHZ state can be created without a \emph{quantum} common cause. We show that this is impossible, even if the parties are granted \emph{classical} shared randomness. To prove this, we derive a family of Bell-like inequalities satisfied by any quantum correlations that can be generated in scenarios where: at most $(N-1)$ parties share a quantum common cause, arbitrary intermediate transformations are allowed and shared randomness is distributed among all $N$ parties. Crucially, for every $N$, these inequalities are violated by suitable local measurements on the GHZ state.

We conclude by discussing the implications of our findings. In particular, we argue that a quantitative version of the Coordination Principle is yet to be found. The inequalities we derive provide only a first step toward the strongest operational constraints implied by the Coordination Principle within quantum information theory. However, we point out the lack of a causal principle that not only forbids \emph{perfect} coordination without a common cause, but also quantitatively bounds how well parties lacking a common cause can coordinate in any reasonable causal theory of information.
 
The present article expands on a more compact companion letter on the same topic~\cite{CompanionPaperPRL} and provides the complete technical proofs. The structure of this manuscript is the following. In Section~\ref{sec:perfect_correlation}, we show that the Coordination principle is satisfied by quantum information theory. Then, in Section~\ref{sec:opt}, we give a brief introduction to the formalism of OPTs and provide the particular OPT that satisfies NSI while violating the Coordination principle. In Section~\ref{fullyquantum_section}, we turn to the genuinely quantum generalization of the coordination task and provide a noise-robust proof to show that the generation of the GHZ state requires a quantum common cause. We conclude in Section~\ref{discussion} with a discussion of our results and potential future directions.

\begin{table*}
{
\begin{tabular}{|c|c|c|c|c|}
\hline
\textbf{Perfect coordination?}                             & Causally Independent & Triangle & One-layer tetrahedron & Two-layer tetrahedron \\ \hline
Reichenbach principle                                     & \xmark     \cite{Reichenbach}                & \cmark & \cmark & \cmark              \\ \hline
\begin{tabular}[c]{@{}c@{}} NSI and Device\\ Replication principles\end{tabular}
 
& \xmark \cite{Reichenbach}                    & \xmark   \cite{Henson2014} & \xmark   \cite{Coiteux2021any}    & \cmark     [See Section~\ref{sec:opt}]         \\ \hline
Quantum Theory                                            & \xmark  \cite{Reichenbach}                   & \xmark    \cite{Henson2014}  & \xmark   \cite{Coiteux2021any}    & \xmark [See Section~\ref{sec:perfect_correlation}]            \\ \hline
\end{tabular}}
\caption{Summary of whether perfect coordination is achievable in various causal scenarios according to different principles or theories.}
\label{Table d}
\end{table*}

\section{Perfect coordination requires a common cause in quantum theory}
\label{sec:perfect_correlation}

In this section, we show that perfect coordination requires a common cause in quantum theory or, equivalently, we prove the following theorem:

\begin{theorem}
Quantum theory satisfies the Coordination Principle.
\label{theoremPRL}
\end{theorem}

Throughout this section, \emph{perfect coordination} refers to the task of generating a shared random bit among the $N$ parties, i.e., the probability distribution
\begin{equation}
    P_{A_1,\dots,A_N}=\frac{1}{2}\big([0\cdots 0]+[1\cdots 1]\big),
    \label{randbit}
\end{equation}
where the bracket notation $[a_1\cdots a_N]$ denotes the event in which party $i$ outputs $a_i$. %

To make the proof of Theorem~\ref{theoremPRL} more accesible, we first treat the case of four parties ($N=4$) in Sec.~\ref{4partitecorrelations}, and then generalize it to the case of an arbitrary number of parties in Sec.~\ref{Npartitecorrelation}. In both cases, we proceed in two steps: first, we provide an analytic proof ruling out \emph{perfect} coordination unless there is a common cause; second, we derive a Bell-like inequality that is robust to noise and can be used as an operational witness.

Before presenting the proofs, we briefly summarize the relevant prior results (see Table~\ref{Table d}). In the trivial scenario of Fig.~\ref{no_correlation}, Reichenbach's principle~\cite{Reichenbach} (or its contrapositive form, the Independence principle) already rules out perfect coordination (actually, it rules out any non-factorizing distribution). Beyond this case, Refs.~\cite{Henson2014,Coiteux2021any} showed, under the NSI principles together with device replication, that perfect coordination is impossible in a broad class of networks in which the observed parties do not share a common cause (e.g., Figs.~\ref{triangle} and \ref{1-layer-tetrahedron} for the cases of $N=3$ and $N=4$, respectively). 

Crucially, however, these results are restricted to causal structures without \emph{intermediate transformations}. It has recently been shown that, for causal models involving nonclassical nodes, allowing intermediate transformations can affect the set of achievable correlations~\cite{centeno2024significance}. Consequently, the above no-go results do not cover the most general causal scenarios when considering nonclassical latent nodes. Here, we close this gap by explicitly allowing intermediate transformations and by analyzing the most general causal structures in which not all parties share a common cause. For $N=4$, the corresponding general causal structure is depicted in Fig.~\ref{2-layer-tetrahedron} and we will refer to it as the \emph{two-layer-tetrahedron} (later, we show that it is indeed the most general such structure).

The notation and terminology that we use throughout the section is that of causal inference. That is, we use directed acyclic graphs (DAGs) to represent causal structures. Within the DAG formalism, there are \emph{observed nodes} that represent the measurements of the parties and \emph{latent nodes} that represent the sources and transformations (namely, quantum states and channels, respectively, in the case of quantum theory).

\subsection{Four-partite perfect coordination requires a common cause in quantum theory}
\label{4partitecorrelations}

Here, we demonstrate that perfect coordination is not achievable in any quantum causal structure with four observed nodes that do not share a common cause. The proof strategy is to first show that there exists a DAG with those features which is the most general one in terms of producible probability distributions and then to prove that perfect coordination is incompatible with it. Thus, let us begin by formally defining the set of DAGs under consideration.

\begin{definition}
We define $\mathcal{G}_4$ as the set of all DAGs with four observed nodes and only classical and/or quantum latent nodes such that the observed nodes do not all share a common cause.
\label{def:setG4}
\end{definition} 

Now, our next step is to show that there exists a particular DAG in $\mathcal{G}_4$ (concretely, the \hyperref[2-layer-tetrahedron]{two-layer tetrahedron}) that can produce any probability distribution which can be obtained in some DAG of $\mathcal{G}_4$. Formally, this is stated as the following lemma:

\begin{lemma}
\label{lemma:multilayerG4}
Any DAG $\textit{g} \in \mathcal{G}_4$ is observationally contained in the \hyperref[2-layer-tetrahedron]{two-layer tetrahedron}.
\end{lemma}

The proof can be found in the Appendix~\ref{app:proof1}. Note that another proof has recently been given using the formalism of lattice theory in Ref.~\cite{van2025order}, but we provide the proof for completeness. Equivalently, Lemma~\ref{lemma:multilayerG4} ensures that if a given probability distribution is infeasible in the \hyperref[2-layer-tetrahedron]{two-layer tetrahedron}, it will be so in any DAG $\textit{g} \in \mathcal{G}_4$. %

We now state a lemma that will serve as an intermediate step in the proof of the main result.

\begin{lemma}
\label{Marcoslemma1}
    Let $|\psi\rangle$ be a pure state and $P$, $Q$ projectors such that
    \begin{itemize}
        \item $\left[ P, Q\right] |\psi\rangle = 0$
        \item $||PQ|\psi\rangle||^2=||P|\psi\rangle||^2=||Q|\psi\rangle||^2$
    \end{itemize}
    Then $P|\psi\rangle = Q|\psi\rangle$.
\end{lemma}

\begin{proof}
Given the assumptions of the lemma, we observe that
    \begin{equation}
    \begin{split}
        \left\langle \psi \left|(P-Q)^2\right|\psi\right\rangle 
        = \left\langle \psi \left| P^2 -2PQ + Q^2 \right|\psi\right\rangle\\
        = \left\langle \psi \left| P \right|\psi\right\rangle  + \left\langle \psi \left| Q \right|\psi\right\rangle -2\left\langle \psi \left| PQ \right|\psi\right\rangle = 0,
    \end{split}
    \end{equation}
    where we have used  the commutativity of $P$ and $Q$ in the first step, the idempotence of projectors in the second, and the second assumption in the last step. Then, it must follow that $\left\langle \psi \left|(P-Q)\right|\psi\right\rangle = 0$ and therefore $P|\psi\rangle = Q|\psi\rangle$.
\end{proof}

After these preliminary steps, we now state and prove our first result for the case of four parties. That is, the four-partite version of Theorem~\ref{theoremPRL}:

\begin{theorem} Consider a DAG $g \in \mathcal{G}_4$, i.e., $g$ has four observed nodes that do not share a common cause and only classical and quantum latent nodes are allowed. Then, perfect coordination over the four observed nodes is not achievable.
\label{theorem:no-perfect-correlation4}
\end{theorem}

\begin{proof}
By Lemma~\ref{lemma:multilayerG4}, we can restrict our analysis to the \hyperref[2-layer-tetrahedron]{two-layer tetrahedron}. Hence, we now show that perfect coordination cannot be achieved in the \hyperref[2-layer-tetrahedron]{two-layer tetrahedron} by contradiction using the quantum inflation technique \cite{wolfe2021quantum} with tailored commutation relations to take into account the presence of intermediate latent nodes \cite{centeno2024significance}.

Note that the probability distribution obtained from the \hyperref[2-layer-tetrahedron]{two-layer tetrahedron} is given by

\begin{equation}
    \begin{split}
P_{A_1,A_2,A_3,A_4}(a_1,a_2,a_3,a_4)= \\ \langle \psi|U^{\dagger} (\Pi_{A_1}^{a_1}\otimes \Pi_{A_2}^{a_2}\otimes \Pi_{A_3}^{a_3}\otimes \Pi_{A_4}^{a_4})U |\psi\rangle.
\end{split}
\end{equation}
where $\Pi_{A_i}^{a_i}$ is the projector that yields outcome $a_i$ of the party $A_i$ ($i \in \{1,2,3,4\}$), $U$ is the unitary associated to the tensor product of all the transformations where each node acts over the appropriate Hilbert spaces, $U = U_{\neg (A_3 A_4)}\otimes U_{\neg (A_2A_4)}\otimes U_{\neg (A_2A_3)}\otimes U_{\neg (A_1A_4)}\otimes U_{\neg(A_1A_3)}\otimes U_{\neg(A_1A_2)}$, and $|\psi\rangle$ is the quantum state defined by the composition of all the sources,  $|\psi\rangle = |\psi\rangle_{\neg (A_4)} \otimes |\psi\rangle_{\neg(A_3)}\otimes |\psi\rangle_{\neg(A_2)}\otimes |\psi\rangle_{\neg(A_1)}$.\footnote{%
Since we do not restrict the Hilbert space dimensions of the latent nodes, the purification theorem and Naimark's dilatation \cite{nielsen2010quantum} allow us, without loss of generality, to take the states to be pure, the operators of the intermediate nodes to be unitaries and the measurements to be projective.}

Now, toward a contradiction, suppose that the four observed nodes of the \hyperref[2-layer-tetrahedron]{two-layer tetrahedron} are perfectly coordinated. That is, there exist states, unitaries and measurements such that $P_{A_1,A_2,A_3,A_4}$ is the distribution of a shared random bit (as per Eq.~\eqref{randbit}).

Then, we construct a quantum inflation using one copy of the observed nodes, $\{A_1,A_2,A_3,A_4\}$, and two copies of each source except for the first and last ones, for which one copy is sufficient.  Hence, the set of sources that we consider is given by $\{\neg (A_1)^{(1)},\neg(A_2)^{(1)},\neg(A_3)^{(1)},\neg(A_4)^{(1)},\neg(A_2)^{(2)},\neg(A_3)^{(2)}\}$ (where the superindex in parentheses labels each copy). Then, to determine how many copies of each intermediate latent node are needed, one needs to assign which copies of the sources are in the causal past of each observed node (that is, specify the set of copy indices of each observed node). We provide such an assignment for our quantum inflation in Table~\ref{tab:quantum_cut4}.\footnote{Note that providing such an assignment is enough to construct a quantum inflation assuming that every observed node's subnetwork replicates the subnetwork of the respective observed node in the original scenario, i.e. every observed node is \emph{individually injectable}. This is the case because each intermediate node in the causal past of a given observed node must share the same copy indices of the sources in its past. For example, $A_1$ has sources $\neg(A_2)^{(1)}, \neg(A_3)^{(1)}$ in its past. Then, it is required to have an intermediate node $\neg(A_2A_3)^{(1,1)}$ with those exact sources in its causal past (as indicated in the superindex).} A graphical representation of the described quantum inflation, hereinafter referred to as \emph{quantum cut inflation}, is presented in Fig.~\ref{fig:inflation-tetrahedron}.

\begin{table}[h]
\center
\begin{tabular}{|c|c|c|c|c|}
\hline
                     & $\neg(A_1)$ & $\neg(A_2)$ & $\neg(A_3)$ & $\neg(A_4)$\\ \hline
$A_1$ & - & 1 & 1 & 1  \\ \hline
$A_2$ & 1 & - & 1 & 1  \\ \hline
$A_3$ & 1 & 2 & - & 1  \\ \hline
$A_4$ & 1 & 2 &  2 & - \\ \hline
\end{tabular}
\caption{Copy indices of the sources (columns) in the causal past of every observed node (rows) of the \emph{quantum cut inflation} which is graphically represented in Fig.~\ref{fig:inflation-tetrahedron}.}
\label{tab:quantum_cut4}
\end{table}

\begin{figure*}
\centering
\scalebox{0.86}{
\begin{tikzpicture}[
  node distance=1cm and 1.5cm,
  every node/.style={draw, circle, minimum size=6mm, inner sep=0pt},
  level 1/.style={sibling distance=15mm},
  level 2/.style={sibling distance=10mm},
  arrow/.style={<-, >=latex}
]
  \definecolor{lightblue}{RGB}{173,216,230}
  \definecolor{lightpink}{RGB}{255,182,193}

  \node[draw, rectangle, minimum width=10mm, minimum height=10mm, font=\large, fill=lightblue] (A) {$A_1$};
  \node[draw, rectangle, minimum width=10mm, minimum height=10mm, font=\large, fill=lightblue, right=0.8cm of A] (B) {$A_2$};
  \node[draw, rectangle, minimum width=10mm, minimum height=10mm, font=\large, fill=lightblue, right= 0.8cm of B] (C) {$A_3$};
  \node[draw, rectangle, minimum width=10mm, minimum height=10mm, font=\large,
  fill=lightblue, right=0.8cm of C] (D) {$A_4$};

  \node[fill=lightpink, below left=1.9cm and 4.65cm of A] (AD1) {$\neg(A_2A_3)^{(1,1)}$};
  \node[fill=lightpink, right=0.3cm of AD1] (AC1) {$\neg(A_2A_4)^{(1,1)}$};
  \node[fill=lightpink, right=0.3cm of AC1] (AB1) {$\neg(A_3A_4)^{(1,1)}$};
  \node[fill=lightpink, right=0.3cm of AB1] (BD1) {$\neg(A_1A_3)^{(1,1)}$};
  \node[fill=lightpink, right=0.3cm of BD1] (BC1) {$\neg(A_1A_4)^{(1,1)}$};
  \node[fill=lightpink, right=0.3cm of BC1] (AC2) {$\neg(A_2A_4)^{(1,2)}$};
  \node[fill=lightpink, right=0.3cm of AC2] (CD1) {$\neg(A_1A_2)^{(2,1)}$};
  \node[fill=lightpink, right=0.3cm of CD1] (BD2) {$\neg(A_1A_3)^{(2,1)}$};
  \node[fill=lightpink, right=0.3cm of BD2] (AD2) {$\neg(A_2A_3)^{(2,2)}$};

  \node[fill=lightpink, below left=6cm and 0.50cm of A] (ACD1) {$\neg(A_2)^{(1)}$};
  \node[fill=lightpink, right=0.5cm of ACD1] (ABD1) {$\neg(A_3)^{(1)}$};
  \node[fill=lightpink, right=0.5cm of ABD1] (ABC1) {$\neg(A_4)^{(1)}$};
  \node[fill=lightpink, right=0.5cm of ABC1] (BCD1) {$\neg(A_1)^{(1)}$};
  \node[fill=lightpink, right=0.5cm of BCD1] (ACD2) {$\neg(A_2)^{(2)}$};
  \node[fill=lightpink, right=0.5cm of ACD2] (ABD2) {$\neg(A_3)^{(2)}$};
  \coordinate (ABCpivot) at (2,-6);
  \fill (ABCpivot) circle (4pt);
  \draw (ABC1) -- (ABCpivot);

  \coordinate (BCDpivot) at (4,-6);
  \fill (BCDpivot) circle (4pt);
  \draw (BCD1) -- (BCDpivot);
  \foreach \s/\d in {A/AD1, A/AC1, A/AB1, B/AB1, B/BC1, B/BD1, C/AC2, C/BC1, C/CD1, D/AD2, D/BD2, D/CD1} {\draw[arrow] (\s) -- (\d);}

  \foreach \s/\d in {AD1/ACD1, AD1/ABD1, AC1/ACD1, AB1/ABD1, AB1/ABC1, BD1/ABD1, BD1/BCDpivot, BC1/ABC1, BC1/BCD1, AC2/ACD2,
  CD1/BCD1, CD1/ACD2,
  BD2/BCDpivot, BD2/ABD2,
  AD2/ACD2, AD2/ABD2, AC1/ABCpivot, AC2/ABCpivot} {\draw[arrow] (\s) -- (\d);}

\end{tikzpicture}
} %
    \caption{Graphical representation of the  \emph{quantum cut inflation} of the \hyperref[2-layer-tetrahedron]{two-layer tetrahedron}. The superindex in parenthesis indicate the copy index of the sources in the past of the nodes. The black dots denote whenever the same Hilbert space of a given source state is pertinent in defining more than one transformation. Accordingly, pairs of transformations related by a black dot (i.e., $\neg (A_1A_3)^{(1,1)}$ and $\neg (A_1A_3)^{(2,1)}$ or $\neg (A_2A_4)^{(1,1)}$ and $\neg (A_2A_4)^{(1,2)}$) cannot be physically implemented simultaneously. That implies incompatibility of their downstream measurements ($A_2$ and $A_4$ or $A_1$ and $A_3$, respectively). Nevertheless, while the whole quantum inflation is not physically valid, certain parts of it remain well defined: in particular, any subnetwork obtained by following only one of the outgoing wires from each black dot is implementable, like the adjacent pairs $(A_1,A_2)$, $(A_2,A_3)$ and $(A_3,A_4)$, or the pair $(A_1,A_4)$.} 
    \label{fig:inflation-tetrahedron}
\end{figure*}
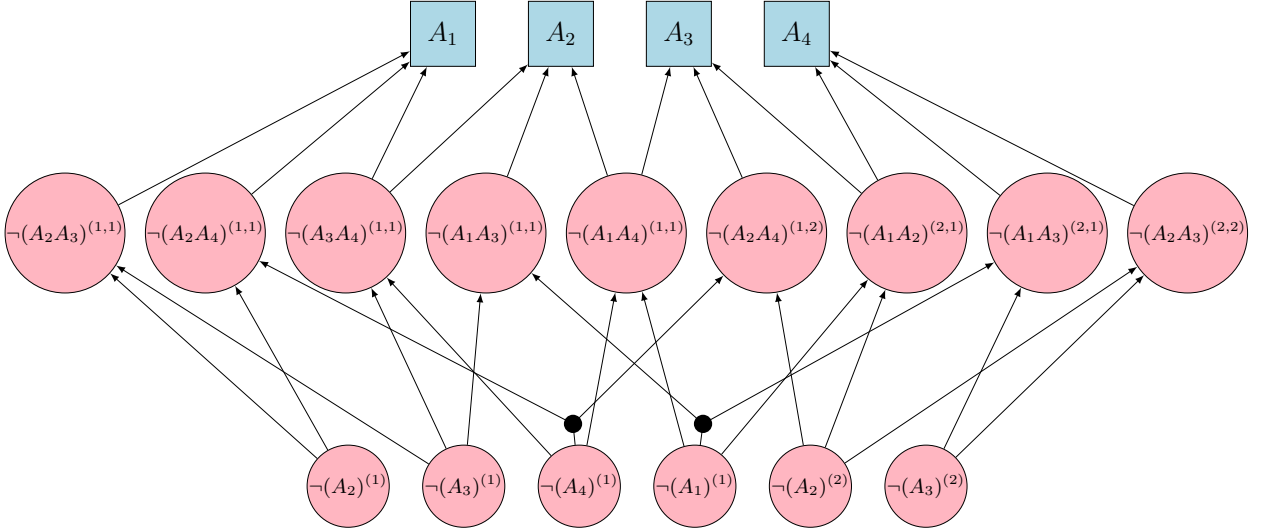

Note that the black dots used in the graphical representation of Fig.~\ref{fig:inflation-tetrahedron} indicate whenever the same Hilbert space of a given source state is pertinent in defining more than one transformation. Consequently, pairs of transformations related by a black dot (i.e., $\neg (A_1A_3)^{(1,1)}$ and $\neg (A_1A_3)^{(2,1)}$ or $\neg (A_2A_4)^{(1,1)}$ and $\neg (A_2A_4)^{(1,2)}$) cannot be physically implemented simultaneously. In turn, this implies incompatibility of their downstream measurements ($A_2$ and $A_4$ or $A_1$ and $A_3$, respectively). Nevertheless, while the whole quantum inflation is not physically valid, certain parts of it remain well defined: in particular, any subnetwork obtained by following only one of the outgoing wires from each black dot is implementable, like the adjacent pairs $(A_1,A_2)$, $(A_2,A_3)$ and $(A_3,A_4)$, or the pair $(A_1,A_4)$.

 The incompatibility of some measurements in the quantum inflation is mathematically implemented by removing commutation rules between the Heisenberg-picture operators describing the measurements of those parties\footnote{As a general rule, notice that, if a pair of observed nodes share the same copy indices for all their common sources, the Heisenberg-picture operators commute, whereas, if the copy indices for the common sources are only partially matched, i.e., some are equal
while others differ, then the corresponding Heisenberg-picture operators do
not commute.} which are defined as

\begin{equation}
    \widetilde{\Pi}_{A_1}^{a_1} \coloneqq U_{A_1}^{\dagger}\Pi_{A_1}^{a_1} U_{A_1}
\end{equation}
where $U_{A_1}$ is the tensor product of all the unitaries corresponding to the transformations in the past of
$A_1$, namely  $U_{A_1}\coloneqq {U_{\neg(A_2A_3)^{(1,1)}} \otimes U_{\neg(A_2A_4)^{(1,1)}}\otimes U_{\neg(A_3A_4)^{(1,1)}}}$. The remaining operators $\widetilde{\Pi}_{B}^{b}$, $\widetilde{\Pi}_{C}^{c}$ and $\widetilde{\Pi}_{D}^{d}$ are defined analogously. In contrast, the Heisenberg-picture operators describing compatible measurements do commute. That is, we have that

\begin{equation}
    \left[ \widetilde{\Pi}_{A_i}^{0},\widetilde{\Pi}_{A_{i+1}}^{0}\right] = 0 \quad\text{for} \quad i=1,2,3.
    \label{commutation}
\end{equation}

Moreover, the two-party sets of observed nodes $(A_1,A_2)$, $(A_2,A_3)$ and $(A_3,A_4)$ are not only compatible but also \emph{injectable}, i.e., their subnetworks in the \emph{quantum cut inflation} exactly replicate the corresponding subnetworks in the \hyperref[2-layer-tetrahedron]{two-layer tetrahedron}. Hence, by consistency between the behaviours in the original scenario and the inflation, we have that their marginals must be those of a shared random bit, i.e.,

\begin{equation}
\begin{split}
\langle \phi | \widetilde{\Pi}_{A_i}^{a_i}\widetilde{\Pi}_{A_{i+1}}^{a_{i+1}=a_i}  |\phi\rangle=\langle \phi | \widetilde{\Pi}_{A_i}^{a_i}  |\phi\rangle \\ 
= \langle \phi | \widetilde{\Pi}_{A_{i+1}}^{a_{i+1}}  |\phi\rangle = \frac{1}{2} \quad\text{for} \quad i=1,2,3
\end{split}
\label{perfect_random}
\end{equation}

Then, as per Eqs.~\eqref{commutation} and \eqref{perfect_random}, we can apply Lemma~\ref{Marcoslemma1} to all adjacent pairs $\left(\widetilde{\Pi}_{A_i}^{0},\widetilde{\Pi}_{A_{i+1}}^{0}\right)$ for $i=1,2,3$,
 
 \begin{equation}
     \widetilde{\Pi}_{A_i}^{0}|\psi\rangle = \widetilde{\Pi}_{A_{i+1}}^{0}|\psi\rangle \quad\text{for} \quad i=1,2,3.
 \end{equation}
 Therefore, by transitivity, we obtain 

\begin{equation}
\label{equality1}
\widetilde{\Pi}_{A_1}^{0}|\psi\rangle = \widetilde{\Pi}_{A_4}^{0}|\psi\rangle.    
\end{equation}

However, from the assignment of the copy indices of the sources (or directly, in the graphical representation of Fig.~\ref{fig:inflation-tetrahedron}), we see that $A_1$ and $A_4$ do not share any common cause in their past in the \emph{quantum cut inflation}, hence they are causally independent (and compatible). This fact translates into the factorization of their marginal probability distribution, i.e. 
\begin{equation}
    \langle\psi|\widetilde{\Pi}_{A_1}^{0} \widetilde{\Pi}_{A_4}^{0}|\psi\rangle = \langle\psi|\widetilde{\Pi}_{A_1}^{0} |\psi\rangle \langle\psi|\widetilde{\Pi}_{A_4}^{0} |\psi\rangle 
\end{equation}
which is in contradiction with Eq.~\eqref{equality1}. Namely, $(A_1,A_4)$
cannot be both perfectly coordinated and independent.    
\end{proof}

As a first step towards the goal of finding the strongest operational constraints implied by
the Coordination Principle in quantum theory, we exploit the quantum inflation method. It allows us to derive the noise-robust counterpart of Theorem~\ref{theorem:no-perfect-correlation4} by providing a Bell-like inequality to witness the presence of a common cause shared by the four parties. Concretely,

\begin{theorem}
    Any correlation produced in a DAG $g \in \mathcal{G}_4$, must satisfy
\begin{equation}
\text{\small
    $\langle A_1A_2 \rangle + \langle A_2A_3 \rangle + \langle A_3A_4 \rangle \leq \sin(\frac{\pi}{6})\langle A_1\rangle \langle  A_4 \rangle + 3\cos(\frac{\pi}{6}).$}
    \label{Ineq_th}
\end{equation}
    \label{noiserobusttheorem4}
\end{theorem}
We defer the formal proof of the previous theorem to Appendix~\ref{SDP_appendix}. However, we provide a proof sketch here.

\textit{Sketch of the proof.} By Lemma~\ref{lemma:multilayerG4}, any inequality %
we derive on the set of probability distributions of the \hyperref[2-layer-tetrahedron]{two-layer tetrahedron} must be satisfied also by any correlation achieved in any DAG $g \in \mathcal{G}_4$.
Hence, we start by considering the same quantum inflation of the \hyperref[2-layer-tetrahedron]{two-layer tetrahedron} as in the previous analytical proof, see Fig.~\ref{fig:inflation-tetrahedron} and Table~\ref{tab:quantum_cut4}. Then, one can use the first level of the NPA hierarchy \cite{navascues2008convergent} to set up a positive semidefinite (PSD) completion problem. That is, one can construct a moment matrix taking into account the injectable sets and the commutation rules described for the operators in the Heisenberg picture which must be positive if the correlation is obtained using a quantum protocol. Using this method, we obtain a certificate of the infeasibility of perfect coordination in the \hyperref[2-layer-tetrahedron]{two-layer tetrahedron}, namely Eq.~\eqref{Ineq_th}.

\hfill\qedsymbol

Notice that, for the particular case of perfect coordination, the left-hand-side of the Eq.~\eqref{Ineq_th} is 3 while the right-hand-side is $\frac{3\sqrt{3}}{2}\approx2.598$ leading to a violation of the inequality. To study the robustness to noise of Eq.~\eqref{Ineq_th} when considering the presence of white noise modeled by the uniform distribution $P_{wn}=\frac{1}{16}\sum_{i \in \{0,1\}^4} [i]$, the distribution is
\begin{equation}
    P = vP_{pc}+(1-v)P_{wn},
    \label{noisypc}
\end{equation}
where $P_{pc}$ corresponds to the ideal scenario of the perfect coordination distribution and $v\in[0,1]$ is the noise parameter. It is easy to check that Eq. \eqref{Ineq_th} allows one to rule out all distributions of the form of Eq.~\eqref{noisypc} in which $v> \cos(\frac{\pi}{6})$. 

Note that the provided Bell-like inequality is not the optimal one in terms of how much noise one could accept to determine the presence of a common cause of the four parties, as it was derive using only the first level of the NPA hierarchy for the \emph{quantum cut inflation}. However, increasing the number of copies or the level of the NPA hierarchy rapidly exceeds the current computational limits.

\subsection{N-partite perfect coordination requires a common cause in quantum theory}
\label{Npartitecorrelation}

We now generalize the results of the previous section to the case of $N$ parties following the same logic, culminating in our first main result, Theorem~\ref{theoremPRL}. Accordingly, let us start by defining the generalization of $\mathcal{G}_4$ to $N$ observed nodes:

\begin{definition}
We define $\mathcal{G}_N$ as the set of all DAGs with $N$ observed nodes and only classical and/or quantum latent nodes such that the observed nodes do not all share a common cause.
\label{def:setN}
\end{definition}

To follow the same steps, we also need to define the analogue of the \hyperref[2-layer-tetrahedron]{two-layer tetrahedron} for $N$ observed nodes.

\begin{definition}\label{def:gNstar}For a given integer $N$, we define the DAG $\textit{g}^*_N$ to be $(V,E)$, where $V$ is the set of all nodes in $\textit{g}^*_N$ and $E \subseteq V \times V$ the set of directed edges, each oriented from the first node in the pair to the second, in the following way:
\begin{itemize}
    \item $V = \mathcal{P}\left(\{A_1, A_2,...,A_N\}\right)\backslash\left\{ \emptyset,\{A_1,A_2,\dots A_N\right\}\}$,
    \item $E = \{(u,v)|\; u,v\in V,\; u\subset v \text{ and }|u|=|v|+1 \}$,
\end{itemize}
where given a set $S$, $\mathcal{P}\left(S\right)$ denotes the power set and $|S|$ the cardinality. Also, we define the nodes whose labels have cardinality one as classical observed nodes, and the rest as quantum latent nodes.
\label{def:DAG}
\end{definition}

For simplicity, we label each node by concatenating the elements of the set it represents (i.e., we omit curly brackets and commas). More concretely, we label latent nodes by the complement of their corresponding set, denoted $\neg(S)$, whereas the observed nodes are simply denoted by $A_1,\dots,A_N$.
Our notation is chosen to make causal relationships explicit. For example, the node $\neg(A_1)$ lies in the causal past of every observed node except $A_1$ itself. See the \hyperref[2-layer-tetrahedron]{two-layer tetrahedron} for the particular case of $N=4$ (i.e., $g_4^*$).

Then, Lemma~\ref{lemma:multilayerG4} generalizes trivially, as the proof (given in Appendix~\ref{app:proof1}) is completely analogous when substituting the number of observed nodes by $N$.

\begin{lemma}
\label{lemma:multilayerG}
Any DAG $\textit{g} \in \mathcal{G}_N$ is observationally contained in $g^*_N$.
\end{lemma}

Let us now state the generalization of Theorem~\ref{theorem:no-perfect-correlation4} to the $N$-partite case. Equivalently, this is just a technical version of Theorem~\ref{theoremPRL}.

\begin{theoremrevisited}{theoremPRL}\label{th1_revisited}
Consider a DAG $g \in \mathcal{G}_N$, i.e., $g$ has N observed nodes that do not share a common cause with only classical and quantum latent nodes. Then, perfect coordination over the N observed nodes is not achievable.
\end{theoremrevisited}

\begin{proof}
This proof is completely analogous to the one of Theorem~\ref{theorem:no-perfect-correlation4} but using Lemma~\ref{lemma:multilayerG} instead of Lemma~\ref{lemma:multilayerG4} and using $g^*_N$ instead of $g_4^*$.

Note that the probability distribution obtained from $g^*_N$ is of the form
\begin{equation}
    \begin{split}
P_{A_1,A_2,\dots,A_N}(a_1,a_2,\dots,a_N)=   \\ \bra{\psi} U_2^{\dagger}\dots U_{N-2}^{\dagger}(\Pi_{A_1}^{a_1}\otimes \cdots \otimes \Pi_{A_N}^{a_N})U_{N-2}\dots U_2 \ket{\psi} ,
\end{split}
\end{equation}
where $\Pi^{A_i}_{a_i}$ is the projector that yields outcome $a_i$ of the party $A_i$,  $U_j$ is the unitary associated to the $j^{th}$ layer of latent nodes (with $j \in \{2,\dots,N-2\}$ since only the intermediate layers are transformations), that is, the tensor product of all the transformations of that layer where each node acts over the appropriate Hilbert spaces, and $|\psi\rangle$ is the initial quantum state defined by the composition of all the sources, i.e., $|\psi\rangle = \neg (A_1) \otimes \cdots \otimes \neg (A_N)$.

Now, to reach a contradiction, we start by assuming that the $N$ observed nodes of $g_N^*$ are perfectly coordinated and then, we use the previous idea of constructing a quantum cut inflation. Hence, we shall use one copy of all the observed nodes, $\{A_1,\dots, A_{N}\}$, and two copies of all the sources except for the first and last ones for which one copy is sufficient, that is, $\{\neg(A_1)^{(1)},\neg (A_2)^{(1)},\dots,\neg A_{N}^{(1)},\neg A_{2}^{(2)},\dots, \neg A_{N-1}^{(2)}\}$. Also, to determine the intermediate latents required to construct the inflation in which all the observed nodes are individually injectable, we provide the copy indices in the past of each observed node in Table~\ref{tab:quantum_cut}.

\begin{table}[h]
    \center
\begin{tabular}{|c|c|c|c|c|c|c|}
\hline
                     & $\neg (A_1)$ & $\neg (A_2)$ & $\neg (A_3)$ & $\neg(A_4)$ & ... & $\neg (A_N)$ \\ \hline
$A_1$ & - & 1 & 1 & 1 & 1   & 1 \\ \hline
$A_2$ & 1 & - & 1 & 1 & 1   & 1 \\ \hline
$A_3$ & 1 & 2 & - & 1 & 1   & 1 \\ \hline
$A_4$ & 1 & 2 & 2 & - & 1 & 1  \\ \hline
...                  & 1 & 2 &  2 & 2 &  -   & 1 \\ \hline
$A_N$ & 1 & 2 & 2 & 2 & 2  & - \\ \hline
\end{tabular}
\caption{Copy indices of the sources (columns) in the causal past of every observed node (rows) of the quantum cut inflation for the case of $N$ observed nodes.}
\label{tab:quantum_cut}
\end{table}

Now, the operators in the Heisenberg picture take into account all the layers of intermediate latent nodes in the past of each observed node. Therefore,
\begin{equation}
    \widetilde{\Pi}_{A_i}^{a_i} = U_{2,A_i}^{\dagger}\dots U_{N-2,A_i}^{\dagger}\Pi_{A_i}^{a_i} U_{N-2,A_i}\dots U_{2,A_i},
\end{equation}
where $i \in \{1,...,N\}$ and $U_{j,A_i}$ is the tensor product of all the unitaries from the $j^{th}$ layer of latent nodes which lie in the causal past of $A_i$.

As in the proof of Theorem~\ref{theorem:no-perfect-correlation4}, our choice of copy indices ensures that the Heisenberg-picture operators associated with all adjacent pairs commute. Moreover, each adjacent pair of observed nodes forms an injectable set, and hence its marginal distribution is constrained to be that of a shared random bit. Together, these two facts imply that Eqs.~\ref{commutation} and~\ref{perfect_random} hold for all $N-1$ adjacent pairs.

Hence, we can apply Lemma~\ref{Marcoslemma1} to all pairs $\left(\widetilde{\Pi}_{A_i}^{0},\widetilde{\Pi}_{A_{i+1}}^{0}\right)$ for $i=1,..., N-1$, obtaining

\begin{equation}
     \widetilde{\Pi}_{A_i}^{0}|\psi\rangle = \widetilde{\Pi}_{A_{i+1}}^{0}|\psi\rangle \quad\text{for} \quad i=1,..., N-1.
\end{equation}

Finally, by transitivity

\begin{equation}
\label{equality2}
\widetilde{\Pi}_{A_1}^{0}|\psi\rangle = \widetilde{\Pi}_{A_N}^{0}|\psi\rangle.    
\end{equation} 

However, from the assignment of the copy indices, we also see that $A_1$ and $A_N$ are causally independent and, thus, 
\begin{equation}
    \langle\psi|\widetilde{\Pi}_{A_1}^{0} \widetilde{\Pi}_{A_N}^{0}|\psi\rangle = \langle\psi|\widetilde{\Pi}_{A_1}^{0} |\psi\rangle \langle\psi|\widetilde{\Pi}_{A_N}^{0} |\psi\rangle 
\end{equation}
which is in contradiction with Eq.~\eqref{equality2}.
    
\end{proof}

Finally, we extend the noise-robust formulation to the general $N$-partite case:

\begin{theoremnoise}{theoremPRL}\label{noiserobusttheorem}
Any correlation produced in a DAG $g \in \mathcal{G}_N$ must satisfy
    
\begin{equation}
\begin{split}
    \sum_{i=1}^{N-1}\langle A_i A_{i+1} \rangle\\ \leq \sin(\frac{\pi}{2(N-1)}) \langle A_1 \rangle \langle A_N \rangle + (N-1)\cos(\frac{\pi}{2(N-1)}).
    \end{split}
    \label{Ineq_thN}
\end{equation}
\end{theoremnoise}

The sketch of the proof is the same as the one for Theorem~\ref{noiserobusttheorem4} but with a larger moment matrix as we have $N$ parties instead of only four. See Appendix~\ref{SDP_appendix} for the detailed proof.

As before, we can verify the noise robustness of this result by introducing noise into the perfect coordination distribution. For $N$ parties, if we assume a white noise model given by $P_{wn}=\frac{1}{2^N}\sum_{i \in \{0,1\}^N} [i]$, the noisy distribution takes the same convex combination form as in Eq.~\eqref{noisypc}. According to Theorem~\nameref{noiserobusttheorem}, we rule out all probability distributions with noise parameter $v>\cos(\frac{\pi}{2(N-1)})$. The tolerance to noise decreases as $N$ increases, eventually vanishing in the limit $N \to \infty$. This behavior is both natural and intuitive: since our method relies on two-party marginals to establish correlations between the first and last parties, a greater number of middle steps makes the overall correlation increasingly sensitive to noise.

\section{A theory producing perfect coordination without a common cause.}
\label{sec:opt}

The framework of Operational Probabilistic Theories (OPTs) provides a language for expressing physical theories through the probabilities they assign to different operations. Within this framework, both classical and quantum theories, as well as beyond-quantum alternatives, can be represented and, hence, systematically compared. This makes OPTs a helpful tool for identifying the principles that distinguish quantum theory from other conceivable physical theories.
In this section, we provide a particular OPT in which multipartite perfect coordination is possible even without sharing a common cause whilst NSI are satisfied. However, before providing its explicit description, we recall the main concepts and definitions used to describe an OPT, which were formally explained in Ref.~\cite{d2017quantum}. Note that in this section, we adopt the terminology of the OPT literature rather than the one of causal inference used in the rest of the manuscript. The necessary connections between the two frameworks are made as new terms are introduced.

\subsection{Definition of OPT}

The primitive notions used to describe an OPT in the language of D'Ariano \emph{et al.} are events, tests and systems. An \emph{event} is a single process that connects an incoming system (the input)
with an outgoing system (the output) indicating the types of system. The notion of system types refers to the different physical inputs and outputs of a physical process like, for example, a spin particle in a Stern-Gerlach experiment. Here, we use different labels to describe system types that are operationally distinct. Then, a \emph{test} is a collection of events. %
Graphically, this is represented as shown in Fig.~\ref{test}, where we label the different system's types with numbers, and the different events of the test with an outcome $a \in |A|$ (where $|A|$ is the cardinality of the set of events in the test). Note that if the test contains a single event, we will not specify any outcome.

\begin{figure}[hb]
    \centering
    \begin{tikzpicture}[
  node distance=1cm and 1.5cm,
  level 1/.style={sibling distance=15mm},
  level 2/.style={sibling distance=10mm},
  arrow/.style={<-, >=latex}
]
        \node[draw, shape=diamond, minimum size=1cm] (1) {$\{A^a\}_{a \in |A|}$};
        \node [style=rectangle] (2) at (0, 2) { };
        \node [style=rectangle] (3) at (0, -2) { };
        \node [style=rectangle] (4) at (0.2, -1.5) {1};
        \node [style=rectangle] (5) at (0.2, 1.4) {2};
        \draw[arrow] (2) -- (1);
        \draw[arrow] (1) -- (3);
    \end{tikzpicture}
    \caption{Diagrammatic representation of a test connecting system $1$ with $2$.}
    \label{test}
\end{figure}
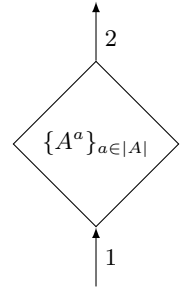

Importantly, any circuit is constructed following a set of compositional rules for the tests. In particular, tests can always be implemented in parallel, regardless of their connectivity. However, sequential composition is permitted only when the output system type of the preceding test matches the input system type of the subsequent test. This is illustrated diagrammatically in Figs.~\ref{parallel},~\ref{sequential}. 
 
\begin{figure}
    \centering
    \begin{tikzpicture}[
  node distance=1cm and 1.5cm,
  level 1/.style={sibling distance=15mm},
  level 2/.style={sibling distance=10mm},
  arrow/.style={<-, >=latex}
]
        \node[draw, shape=diamond, minimum size=1cm] (1) {$\{A^a\}_{a \in |A|}$};
        \node [style=rectangle] (2) at (0, 2) { };
        \node [style=rectangle] (3) at (0, -2) { };
        \node [style=rectangle] (4) at (0.2, -1.5) {1};
        \node [style=rectangle] (5) at (0.2, 1.6) {2};
        \draw[arrow] (2) -- (1);
        \draw[arrow] (1) -- (3);
        \node[draw, shape=diamond, minimum size=1cm] (6) at (3,0) {$\{B^b\}_{b \in |B|}$};
        \node [style=rectangle] (7) at (3, 2) { };
        \node [style=rectangle] (8) at (3, -2) { };
        \node [style=rectangle] (9) at (3.2, -1.5) {2};
        \node [style=rectangle] (10) at (3.2, 1.6) {3};
        \draw[arrow] (7) -- (6);
        \draw[arrow] (6) -- (8);
        \node[draw, style=rectangle, dashed, minimum width=5.5cm, minimum height=2.7cm] (11) at (1.5,0) {};
    \end{tikzpicture}
    \caption{Diagrammatic representation of parallel composition. %
    }
    \label{parallel}
\end{figure}
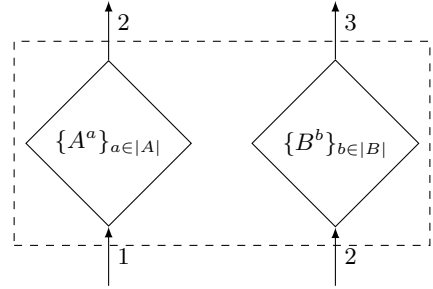

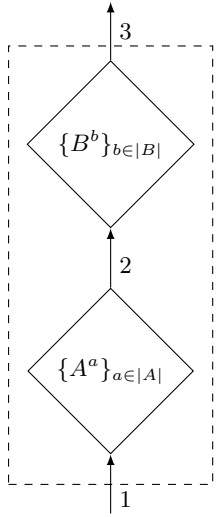
\begin{figure}
    \centering
    \begin{tikzpicture}[
  node distance=1cm and 1.5cm,
  level 1/.style={sibling distance=15mm},
  level 2/.style={sibling distance=10mm},
  arrow/.style={<-, >=latex}
]
        \node[draw, shape=diamond, minimum size=1cm] (1) {$\{A^a\}_{a \in |A|}$};
        \node [style=rectangle] (2) at (0, 5) { };
        \node [style=rectangle] (3) at (0, -2) { };
        \node [style=rectangle] (4) at (0.2, -1.7) {1};
        \node [style=rectangle] (5) at (0.2, 1.4) {2};
        \node[draw, shape=diamond, minimum size=1cm] (6) at (0,3) {$\{B^b\}_{b \in |B|}$};
        \node [style=rectangle] (10) at (0.2,4.5) {3};
        \draw[arrow] (2) -- (6);
        \draw[arrow] (1) -- (3);
        \draw[arrow] (6) -- (1);
        \node[draw, style=rectangle, dashed, minimum width=2.7cm, minimum height=5.8cm] (11) at (0,1.4) {};
    \end{tikzpicture}
    \caption{Diagrammatic representation of sequential composition.}
    \label{sequential}
\end{figure}

Some tests have either no systems as input or no systems as output \footnote{The notion of a test ``not having'' an input/output is formally stated as that their input/output is the trivial system.}. We refer to these as \emph{preparation tests} and \emph{observation tests}\footnote{Note that in the notation used in the rest of the paper, observations and preparations correspond to observed classical nodes (parties) and sources, respectively.}, respectively. In Fig.~\ref{circuit}, the test $\{A^a\}_{a \in |A|}$ represents a preparation, while $\{C^c\}_{c \in |C|}$ represents an observation. All other tests will be referred to as \emph{transformations}\footnote{Note that, here we are using the term transformation in a different way of how it is used in Ref.\cite{d2017quantum}. Formally, we use the notion of transformation for the tests whose input and output systems are not the trivial system.}. Of special importance are \emph{closed circuits} which can  be viewed as a test without input and output. %
Therefore, the OPT must assign a valid joint probability (i.e., positive and normalized) to every combination of events in such closed circuits %
. For instance, in the circuit shown in Fig.~\ref{circuit}, the theory assigns a joint probability distribution $p(a,b,c)$ over all possible events. The assignment of all these probabilities is known as a \emph{probability rule}.

\begin{figure}
    \centering
    \begin{tikzpicture}[
  node distance=1cm and 1.5cm,
  level 1/.style={sibling distance=15mm},
  level 2/.style={sibling distance=10mm},
  arrow/.style={<-, >=latex}
]
        \node[draw, shape=semicircle, shape border rotate=180] (1) {$\{A^a\}_{a \in |A|}$};
        \node [style=rectangle] (5) at (0.2, 0.8) {2};
        \node[draw, shape=diamond, minimum size=1cm] (6) at (0,2.5) {$\{B^b\}_{b \in |B|}$};
        \node [style=rectangle] (10) at (0.2,4) {3};
        \node[draw, shape=semicircle] (11) at (0,5) {$\{C^c\}_{c \in |C|}$};
        \draw[arrow] (11) -- (6);
        \draw[arrow] (6) -- (1);
    \end{tikzpicture}
    \caption{Closed circuit with a preparation, a transformation and an observation.}
    \label{circuit}
\end{figure}
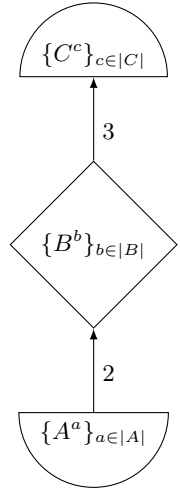

We also recall the notion of causality as defined in Ref.~\cite{d2017quantum}. The principle of causality corresponds to the idea of ``no signaling from the future'' or ``no signaling without interaction''. In physical terms, the outcome of a test can only depend on events within its past light cone. In the OPT framework, this principle is mathematically formalized as follows: causality requires that the probabilities of preparations' outcomes are independent of the choice of observations. That is, the OPT satisfies the causality principle of Ref.~\cite{d2017quantum} if the marginal distribution of any preparation do not change depending on the observation used to close the circuit. Remarkably, this notion of causality is equivalent to the existence of a unique deterministic effect\footnote{The generalized probability theories (GPT) formalism call this the unit effect.} (Lemma 5.1 in Ref.~\cite{d2017quantum}), $e_i$ for any system $i$. A deterministic effect is defined simply as the sum of all the events of an observation test; therefore, we say that an OPT is causal if for all system $i$,   $e_i=\sum_{a \in |A|} A^a \;\; \forall \; \{A^a\}$ (where $\{A^a\}$ constitutes an observation test of system $i$). As we show later, there are OPTs that satisfy this notion of causality while violating the Coordination principle. Therefore, we see this notion of causality as a necessary condition, but not sufficient to say that a theory is causal.

\subsection{Providing the concrete OPT}

We now provide the explicit OPT that allows one to achieve perfect coordination over four parties which share no common cause while satisfying NSI. Let us present it in two steps. First, define the different available tests, that is, the available preparations, transformations and observations, and, secondly, the probability rule to assign a valid probability to any closed circuit we can construct using the available tests.

\begin{definition}
    The available tests of the OPT $\mathcal{T}$, are:
    \begin{itemize}
        \item Four different tripartite preparations, see Fig.~\ref{preparations}.
        \item Six different two input-two output transformations, see Fig.~\ref{transformations}.
        \item Four different tripartite two-outcome observations, see Fig.~\ref{observations}.
    \end{itemize}
\end{definition}

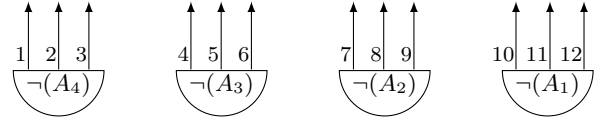
\begin{figure}
    \centering
    \begin{tikzpicture}[
  node distance=1cm and 1.5cm,
  every node/.style={draw, circle, minimum size=6mm, inner sep=0pt},
  level 1/.style={sibling distance=15mm},
  level 2/.style={sibling distance=10mm},
  arrow/.style={<-, >=latex}
]
  \definecolor{lightblue}{RGB}{173,216,230}
  \definecolor{lightpink}{RGB}{255,182,193}

  \node[shape=semicircle, shape border rotate=180] (ABC) {$ \neg (A_4)$};
  \node[shape=semicircle, right=1cm of ABC,shape border rotate=180] (ABD) {$\neg (A_3)$};
  \node[shape=semicircle, right=1cm of ABD, shape border rotate=180] (ACD) {$\neg (A_2)$};
  \node[shape=semicircle, right=1cm of ACD, shape border rotate=180] (BCD) {$\neg (A_1)$};
  \node[style=rectangle, draw=none] (1) at (-0.5, 0.4) {\footnotesize 1};
  \node[style=rectangle, draw=none] (2) at (-0.1, 0.4) {\footnotesize 2};
  \node[style=rectangle, draw=none] (3) at (0.3, 0.4) {\footnotesize 3};
  \node[style=rectangle, draw=none] (4) at (1.65, 0.4) {\footnotesize 4};
  \node[style=rectangle, draw=none] (5) at (2.05, 0.4) {\footnotesize 5};
  \node[style=rectangle, draw=none] (6) at (2.45, 0.4) {\footnotesize 6};
  \node[style=rectangle, draw=none] (7) at (3.8, 0.4) {\footnotesize 7};
  \node[style=rectangle, draw=none] (8) at (4.2, 0.4) {\footnotesize 8};
  \node[style=rectangle, draw=none] (9) at (4.6, 0.4) {\footnotesize 9};
  \node[style=rectangle, draw=none] (10) at (5.88, 0.4) {\footnotesize 10};
  \node[style=rectangle, draw=none] (11) at (6.33, 0.4) {\footnotesize 11};
  \node[style=rectangle, draw=none] (12) at (6.78, 0.4) {\footnotesize 12};
  \draw[arrow] (-0.4, 1.1) -- (-0.4, 0.2);
  \draw[arrow] (-0, 1.1) -- (-0, 0.2);
  \draw[arrow] (0.4, 1.1) -- (0.4, 0.2);
  \draw[arrow] (1.75, 1.1) -- (1.75, 0.2);
  \draw[arrow] (2.15, 1.1) -- (2.15, 0.2);
  \draw[arrow] (2.55, 1.1) -- (2.55, 0.2);
  \draw[arrow] (3.9, 1.1) -- (3.9, 0.2);
  \draw[arrow] (4.3, 1.1) -- (4.3, 0.2);
  \draw[arrow] (4.7, 1.1) -- (4.7, 0.2);
  \draw[arrow] (6.05, 1.1) -- (6.05, 0.2);
  \draw[arrow] (6.5, 1.1) -- (6.5, 0.2);
  \draw[arrow] (6.95, 1.1) -- (6.95, 0.2);

\end{tikzpicture}
    \caption{Different preparations available within the OPT $\mathcal{T}$. Different types of systems are labeled by a number next to the wire.}
    \label{preparations}
\end{figure}

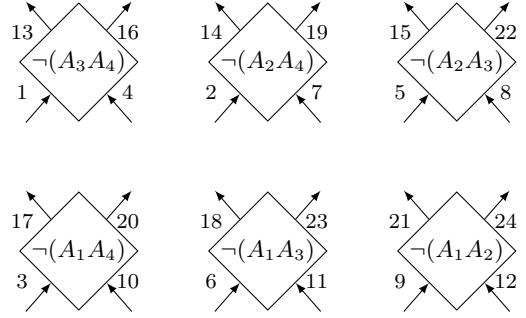
\begin{figure}
    \centering
    \begin{tikzpicture}[
  node distance=1cm and 1.5cm,
  every node/.style={draw, circle, minimum size=6mm, inner sep=0pt},
  level 1/.style={sibling distance=15mm},
  level 2/.style={sibling distance=10mm},
  arrow/.style={<-, >=latex}
]
  \definecolor{lightblue}{RGB}{173,216,230}
  \definecolor{lightpink}{RGB}{255,182,193}

  \node[shape=diamond] (AB) {$ \neg (A_3A_4)$};
  \node[shape=diamond] at (2.5, 0) (AC) {$\neg (A_2A_4)$};
  \node[shape=diamond] at (5, 0) (AD) {$\neg (A_2A_3)$};
  \node[shape=diamond] at (0.0, -2.5) (BC) {$\neg (A_1A_4)$};
  \node[shape=diamond] at (2.5, -2.5) (BD) {$\neg (A_1A_3)$};
  \node[shape=diamond] at (5, -2.5) (CD) {$\neg (A_1A_2)$};
  \node[style=rectangle, draw=none] (13) at (-0.75, 0.4) {\footnotesize 13};
  \node[style=rectangle, draw=none] (16) at (0.65, 0.4) {\footnotesize 16};
  \node[style=rectangle, draw=none] (14) at (1.75, 0.4) {\footnotesize 14};
  \node[style=rectangle, draw=none] (19) at (3.15, 0.4) {\footnotesize 19};
  \node[style=rectangle, draw=none] (15) at (4.25, 0.4) {\footnotesize 15};
  \node[style=rectangle, draw=none] (22) at (5.65, 0.4) {\footnotesize 22};
  \node[style=rectangle, draw=none] (1) at (-0.75, -0.4) {\footnotesize 1};
  \node[style=rectangle, draw=none] (4) at (0.65, -0.4) {\footnotesize 4};
  \node[style=rectangle, draw=none] (2) at (1.75, -0.4) {\footnotesize 2};
  \node[style=rectangle, draw=none] (7) at (3.15, -0.4) {\footnotesize 7};
  \node[style=rectangle, draw=none] (5) at (4.25, -0.4) {\footnotesize 5};
  \node[style=rectangle, draw=none] (8) at (5.65, -0.4) {\footnotesize 8};
  \node[style=rectangle, draw=none] (17) at (-0.75, -2.1) {\footnotesize 17};
  \node[style=rectangle, draw=none] (20) at (0.65, -2.1) {\footnotesize 20};
  \node[style=rectangle, draw=none] (18) at (1.75, -2.1) {\footnotesize 18};
  \node[style=rectangle, draw=none] (23) at (3.15, -2.1) {\footnotesize 23};
  \node[style=rectangle, draw=none] (21) at (4.25, -2.1) {\footnotesize 21};
  \node[style=rectangle, draw=none] (24) at (5.65, -2.1) {\footnotesize 24};
  \node[style=rectangle, draw=none] (3) at (-0.75, -2.9) {\footnotesize 3};
  \node[style=rectangle, draw=none] (10) at (0.65, -2.9) {\footnotesize 10};
  \node[style=rectangle, draw=none] (6) at (1.75, -2.9) {\footnotesize 6};
  \node[style=rectangle, draw=none] (11) at (3.15, -2.9) {\footnotesize 11};
  \node[style=rectangle, draw=none] (9) at (4.25, -2.9) {\footnotesize 9};
  \node[style=rectangle, draw=none] (12) at (5.65, -2.9) {\footnotesize 12};
  \draw[arrow] (-0.7, 0.8) -- (AB);
  \draw[arrow] (0.7, 0.8) -- (AB);
  \draw[arrow] (AB) -- (-0.7, -0.8);
  \draw[arrow] (AB) -- (0.7, -0.8);
  \draw[arrow] (-0.7, -1.7) -- (BC);
  \draw[arrow] (0.7, -1.7) -- (BC);
  \draw[arrow] (BC) -- (-0.7, -3.3);
  \draw[arrow] (BC) -- (0.7, -3.3);
  \draw[arrow] (1.8, 0.8) -- (AC);
  \draw[arrow] (3.2, 0.8) -- (AC);
  \draw[arrow] (AC) -- (1.8, -0.8);
  \draw[arrow] (AC) -- (3.2, -0.8);
  \draw[arrow] (4.3, 0.8) -- (AD);
  \draw[arrow] (5.7, 0.8) -- (AD);
  \draw[arrow] (AD) -- (4.3, -0.8);
  \draw[arrow] (AD) -- (5.7, -0.8);
  \draw[arrow] (1.8, -1.7) -- (BD);
  \draw[arrow] (3.2, -1.7) -- (BD);
  \draw[arrow] (BD) -- (1.8, -3.3);
  \draw[arrow] (BD) -- (3.2, -3.3);
  \draw[arrow] (4.3, -1.7) -- (CD);
  \draw[arrow] (5.7, -1.7) -- (CD);
  \draw[arrow] (CD) -- (4.3, -3.3);
  \draw[arrow] (CD) -- (5.7, -3.3);

\end{tikzpicture}
    \caption{Different transformations available within the OPT $\mathcal{T}$. Different types of systems are labeled by a number next to the wire.}
    \label{transformations}
\end{figure}

\begin{figure}
    \centering
    \begin{tikzpicture}[
  node distance=1cm and 1.5cm,
  every node/.style={draw, circle, minimum size=6mm, inner sep=0pt},
  level 1/.style={sibling distance=15mm},
  level 2/.style={sibling distance=10mm},
  arrow/.style={<-, >=latex}
]
  \definecolor{lightblue}{RGB}{173,216,230}
  \definecolor{lightpink}{RGB}{255,182,193}

  \node[draw, shape=semicircle] at (0, 1.35) (A) {$A_1^{a_1}$};
  \node[draw, shape=semicircle, right=1.cm of A] (B) {$A_2^{a_2}$};
  \node[draw, shape=semicircle, right= 1.05cm of B] (C) {$A_3^{a_3}$};
  \node[draw, shape=semicircle, right=1.1cm of C] (D) {$A_4^{a_4}$};
  \node[style=rectangle, draw=none] (1) at (-0.55, 0.55) {\footnotesize 13};
  \node[style=rectangle, draw=none] (2) at (-0.15, 0.55) {\footnotesize 14};
  \node[style=rectangle, draw=none] (3) at (0.25, 0.55) {\footnotesize 15};
  \node[style=rectangle, draw=none] (4) at (1.6, 0.55) {\footnotesize 16};
  \node[style=rectangle, draw=none] (5) at (2.0, 0.55) {\footnotesize 17};
  \node[style=rectangle, draw=none] (6) at (2.4, 0.55) {\footnotesize 18};
  \node[style=rectangle, draw=none] (7) at (3.75, 0.55) {\footnotesize 19};
  \node[style=rectangle, draw=none] (8) at (4.15, 0.55) {\footnotesize 20};
  \node[style=rectangle, draw=none] (9) at (4.55, 0.55) {\footnotesize 21};
  \node[style=rectangle, draw=none] (10) at (5.88, 0.55) {\footnotesize 22};
  \node[style=rectangle, draw=none] (11) at (6.33, 0.55) {\footnotesize 23};
  \node[style=rectangle, draw=none] (12) at (6.78, 0.55) {\footnotesize 24};
  \draw[arrow] (-0.4, 1.1) -- (-0.4, 0.2);
  \draw[arrow] (-0, 1.1) -- (-0, 0.2);
  \draw[arrow] (0.4, 1.1) -- (0.4, 0.2);
  \draw[arrow] (1.75, 1.1) -- (1.75, 0.2);
  \draw[arrow] (2.15, 1.1) -- (2.15, 0.2);
  \draw[arrow] (2.55, 1.1) -- (2.55, 0.2);
  \draw[arrow] (3.9, 1.1) -- (3.9, 0.2);
  \draw[arrow] (4.3, 1.1) -- (4.3, 0.2);
  \draw[arrow] (4.7, 1.1) -- (4.7, 0.2);
  \draw[arrow] (6.05, 1.1) -- (6.05, 0.2);
  \draw[arrow] (6.5, 1.1) -- (6.5, 0.2);
  \draw[arrow] (6.95, 1.1) -- (6.95, 0.2);

\end{tikzpicture}
    \caption{Different observations available within the OPT $\mathcal{T}$. Different types of systems are labeled by a number next to the wire. Recall all observations are binary, $a_i \in \{0,1\}$.}
    \label{observations}
\end{figure}
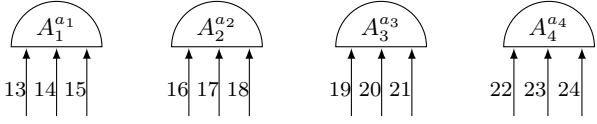

Note that when we say that there are four different tripartite preparations we refer to preparations that yield different types of systems (analogously for transformations and measurements). For simplicity, we assign a number to each different type of system. Connecting the available tests, one can construct different closed circuits. One example is shown in Fig.~\ref{fig:circuit-types}, where all the available tests of the theory are used. Remarkably, such causal structure coincides with that of the \hyperref[2-layer-tetrahedron]{two-layer tetrahedron}. Tracing out\footnote{We use the term ``trace out'' to refer to the fact that we are not interested in the value obtained in the observation test of certain group of systems. We also refer to this as directly tracing out the systems.%
} some of the observations in that circuit, we also obtain closed circuits. We refer to a set of observations that has the same causal past as in the \hyperref[2-layer-tetrahedron]{two-layer tetrahedron} as a \emph{tetrahedron-embeddable} set. For example, one case with only two observations that are jointly tetrahedron-embeddable is illustrated in Fig.~\ref{fig:S1}. Aditionally, closed circuits where the observations have a different causal past than the one in the \hyperref[2-layer-tetrahedron]{two-layer tetrahedron} (i.e., observations which are \emph{not tetrahedron-embeddable}) are also possible within this OPT. One such example is represented in Fig.~\ref{noninjectable_singleton}, where we only have one observation test that has two different preparations of the same type in its causal past.

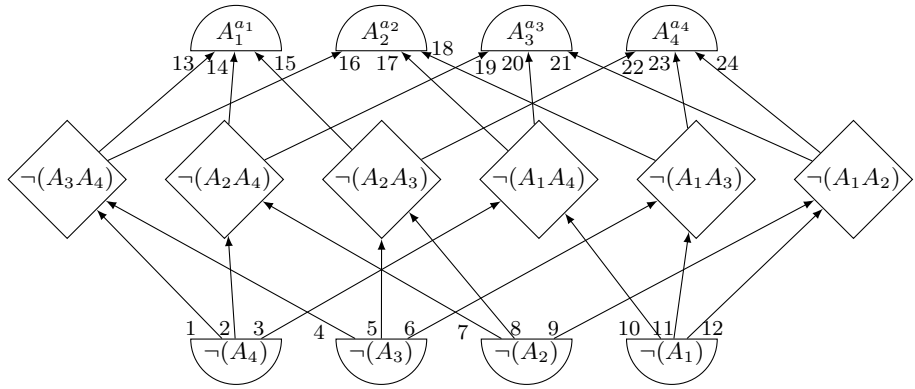
\begin{figure*}
    \centering
    \begin{tikzpicture}[
  node distance=1cm and 1.5cm,
  every node/.style={draw, circle, minimum size=6mm, inner sep=0pt},
  level 1/.style={sibling distance=15mm},
  level 2/.style={sibling distance=10mm},
  arrow/.style={<-, >=latex}
]
  \definecolor{lightblue}{RGB}{173,216,230}
  \definecolor{lightpink}{RGB}{255,182,193}

  \node[draw, shape=semicircle] (A) {$A_1^{a_1}$};
  \node[draw, shape=semicircle, right=0.8cm of A] (B) {$A_2^{a_2}$};
  \node[draw, shape=semicircle, right= 0.8cm of B] (C) {$A_3^{a_3}$};
  \node[draw, shape=semicircle, right=0.8cm of C] (D) {$A_4^{a_4}$};

  \node[shape=diamond, below left=1.3cm and 1.6cm of A] (AB) {$ \neg (A_3A_4)$};
  \node[shape=diamond, right=0.5cm of AB] (AC) {$\neg (A_2A_4)$};
  \node[shape=diamond, right=0.5cm of AC] (AD) {$\neg (A_2A_3)$};
  \node[shape=diamond, right=0.5cm of AD] (BC) {$\neg (A_1A_4)$};
  \node[shape=diamond, right=0.5cm of BC] (BD) {$\neg (A_1A_3)$};
  \node[shape=diamond,  right=0.5cm of BD] (CD) {$\neg (A_1A_2)$};

  \node[shape=semicircle, below=3.8cm of A, shape border rotate=180] (ABC) {$ \neg (A_4)$};
  \node[shape=semicircle, below=3.8cm of B,shape border rotate=180] (ABD) {$\neg (A_3)$};
  \node[shape=semicircle, below=3.8cm of C, shape border rotate=180] (ACD) {$\neg (A_2)$};
  \node[shape=semicircle, below=3.8cm of D, shape border rotate=180] (BCD) {$\neg (A_1)$};
  \node[style=rectangle, draw=none] (1) at (-0.6, -3.9) {\footnotesize 1};
  \node [style=rectangle, draw=none] (2) at (-0.15, -3.9) {\footnotesize 2};
  \node [style=rectangle, draw=none] (3) at (0.3, -3.9) {\footnotesize 3};
  \node [style=rectangle, draw=none] (4) at (1.1, -3.95) {\footnotesize 4};
  \node [style=rectangle, draw=none] (5) at (1.8, -3.9) {\footnotesize 5};
  \node [style=rectangle, draw=none] (6) at (2.3, -3.9) {\footnotesize 6};
  \node [style=rectangle, draw=none] (7) at (3, -3.95) {\footnotesize 7};
  \node [style=rectangle, draw=none] (8) at (3.7, -3.9) {\footnotesize 8};
  \node [style=rectangle, draw=none] (9) at (4.2, -3.9) {\footnotesize 9};
  \node [style=rectangle, draw=none] (10) at (5.2, -3.9) {\footnotesize 10};
  \node [style=rectangle, draw=none] (11) at (5.65, -3.9) {\footnotesize 11};
  \node [style=rectangle, draw=none] (12) at (6.3, -3.9) {\footnotesize 12};
  \node[style=rectangle, draw=none] (13) at (-0.7, -0.4) {\footnotesize 13};
  \node [style=rectangle, draw=none] (14) at (-0.25, -0.45) {\footnotesize 14};
  \node [style=rectangle, draw=none] (15) at (0.65, -0.4) {\footnotesize 15};
  \node [style=rectangle, draw=none] (16) at (1.5, -0.4) {\footnotesize 16};
  \node [style=rectangle, draw=none] (17) at (2, -0.4) {\footnotesize 17};
  \node [style=rectangle, draw=none] (18) at (2.73, -0.2) {\footnotesize 18};
  \node [style=rectangle, draw=none] (19) at (3.3, -0.45) {\footnotesize 19};
  \node [style=rectangle, draw=none] (20) at (3.66, -0.4) {\footnotesize 20};
  \node [style=rectangle, draw=none] (21) at (4.3, -0.4) {\footnotesize 21};
  \node [style=rectangle, draw=none] (22) at (5.25, -0.45) {\footnotesize 22};
  \node [style=rectangle, draw=none] (23) at (5.6, -0.4) {\footnotesize 23};
  \node [style=rectangle, draw=none] (24) at (6.5
  , -0.4) {\footnotesize 24};
  \foreach \s/\d in {A/AB, A/AC, A/AD, B/AB, B/BC, B/BD, C/AC, C/BC, C/CD, D/AD, D/BD, D/CD} {\draw[arrow] (\s) -- (\d);}
  \foreach \s/\d in {AB/ABC, AB/ABD, AC/ABC, AC/ACD, AD/ABD, AD/ACD, BC/ABC, BC/BCD, BD/ABD, BD/BCD, CD/ACD, CD/BCD} {\draw[arrow] (\s) -- (\d);}

\end{tikzpicture}
    \caption{Closed circuit using all components of the OPT. Note that it presents the same causal structure as the one described by the \hyperref[2-layer-tetrahedron]{two-layer tetrahedron}.}
    \label{fig:circuit-types}
\end{figure*}

\begin{figure*}
    \centering
\begin{tikzpicture}[
  node distance=1cm and 1.5cm,
  every node/.style={draw, circle, minimum size=6mm, inner sep=0pt},
  level 1/.style={sibling distance=15mm},
  level 2/.style={sibling distance=10mm},
  arrow/.style={<-, >=latex}
]
  \definecolor{lightblue}{RGB}{173,216,230}
  \definecolor{lightpink}{RGB}{255,182,193}

  \node[draw, shape=semicircle] (A) {$A_1^{a_1}$};
  \node[draw, shape=semicircle, right=0.8cm of A] (B) {$A_2^{a_2}$};

  \node[shape=diamond, below left=1.3cm and 2.5cm of A] (AC) {$\neg (A_2A_4)$};
  \node[shape=diamond, right=0.5cm of AC] (AD) {$\neg (A_2A_3)$};
  \node[shape=diamond, right=0.5cm of AD] (AB) {$\neg (A_3A_4)$};
  \node[shape=diamond, right=0.5cm of AB] (BC) {$\neg (A_1A_4)$};
  \node[shape=diamond, right=0.5cm of BC] (BD) {$\neg (A_1A_3)$};

  \node[shape=semicircle, below=3.8cm of A, shape border rotate=180] (ABD) {$\neg (A_3)$};
  \node[shape=semicircle,
   left=0.8cm of ABD, shape border rotate=180] (ACD) {$\neg (A_2)$};
  \node[shape=semicircle, right=0.8cm of ABD, shape border rotate=180] (ABC) {$\neg (A_4)$};
  \node[shape=semicircle, right=0.8cm of ABC, shape border rotate=180] (BCD) {$\neg (A_1)$};

  \foreach \s/\d in {A/AB, A/AC, A/AD, B/AB, B/BC, B/BD} {\draw[arrow] (\s) -- (\d);}

  \foreach \s/\d in {AB/ABC, AB/ABD, AC/ABC, AC/ACD, AD/ABD, AD/ACD, BC/ABC, BC/BCD, BD/ABD, BD/BCD} {\draw[arrow] (\s) -- (\d);}

\end{tikzpicture}
    \caption{Closed circuit of a pair of observations that are tetrahedron-embeddable. Traced out systems and type numbers are omitted for simplicity}
    \label{fig:S1}
\end{figure*}
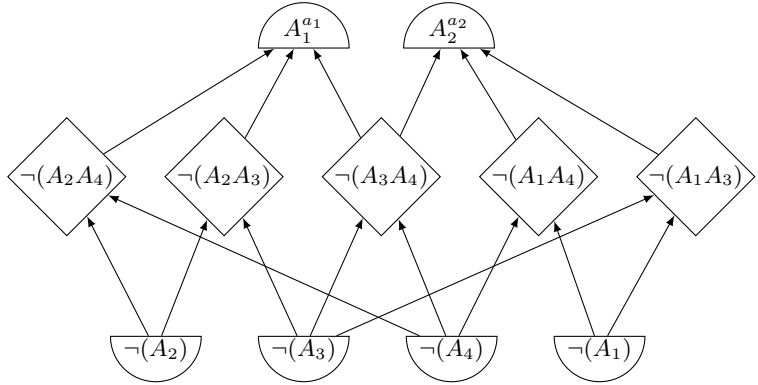

\begin{figure*}
    \centering
\begin{tikzpicture}[
  node distance=1cm and 1.5cm,
  every node/.style={draw, circle, minimum size=6mm, inner sep=0pt},
  level 1/.style={sibling distance=15mm},
  level 2/.style={sibling distance=10mm},
  arrow/.style={<-, >=latex}
]
  \definecolor{lightblue}{RGB}{173,216,230}
  \definecolor{lightpink}{RGB}{255,182,193}

  \node[draw, shape=semicircle] (A) {$A_1^{a_1}$};

  \node[shape=diamond, below=1.3cm of A] (AD) {$\neg (A_2A_3)$};
   \node[shape=diamond, left=0.5cm of AD] (AC) {$\neg (A_2A_4)$};
  \node[shape=diamond, right=0.5cm of AD] (AB) {$\neg (A_3A_4)$};

  \node[shape=semicircle, below right=3.8cm and 0.5cm of A, shape border rotate=180] (ABD) {$\neg (A_3)$};
  \node[shape=semicircle,
   left=0.8cm of ABD, shape border rotate=180] (ACD) {$\neg (A_2)$};
  \node[shape=semicircle, right=0.8cm of ABD, shape border rotate=180] (ABC) {$\neg (A_4)$};
  \node[shape=semicircle, left=0.8cm of ACD, shape border rotate=180] (ABC2) {$\neg (A_4)$};

  \foreach \s/\d in {A/AB, A/AC, A/AD} {\draw[arrow] (\s) -- (\d);}

  \foreach \s/\d in {AB/ABC, AB/ABD, AC/ABC2, AC/ACD, AD/ABD, AD/ACD} {\draw[arrow] (\s) -- (\d);}

\end{tikzpicture}
    \caption{A closed circuit for a single observation whose causal past differs from that in the \hyperref[2-layer-tetrahedron]{two-layer tetrahedron}, i.e., which is not tetrahedron-embeddable. The possibility of such non–tetrahedron-embeddable single observations arises because, in the \hyperref[2-layer-tetrahedron]{two-layer tetrahedron}, each preparation can influence an observation via two distinct paths. Traced out systems and type numbers are omitted for simplicity.}
\label{noninjectable_singleton}
\end{figure*}
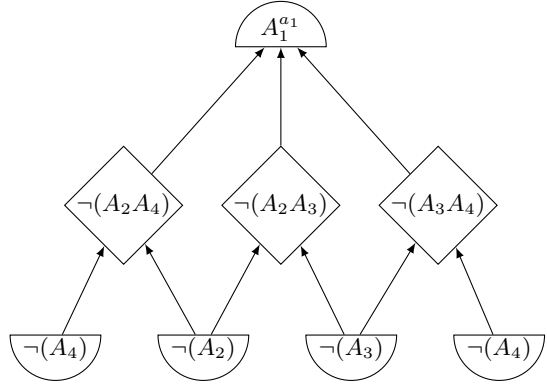

Importantly, the set of closed circuits that can be constructed within this OPT admits a precise characterization: every closed circuit can contain any number of independent sets of observations (i.e., sets that do not share common past). These sets may contain some observations that are jointly tetrahedron-embeddable and/or some observations that are not tetrahedron-embeddable (even when considered individually). 

To prove that all closed circuits are of the previously specified form, it suffices to show that there is no possibility of constructing a closed circuit which contains non-independent sets (of observations) that are tetrahedron-embeddable separately while not being jointly tetrahedron-embeddable. In other words, we need to prove that any two non-independent sets of tetrahedron-embeddable observations must be themselves jointly tetrahedron-embeddable.

We show it in two steps: first, we consider the most elemental case of having two single observations that are tetrahedron-embeddable and not independent and second, we generalize to sets of larger size.

\begin{proposition}
    Within $\mathcal{T}$, if two different observations are not independent (in that they share a common cause) and both are individually tetrahedron-embeddable, then they must be jointly tetrahedron-embeddable.
    \label{singletons_proposition}
\end{proposition}
\begin{proof}
We first note that one cannot construct a closed circuit containing two tetrahedron-embeddable observations of the same type, say $A_i$, whose causal pasts are distinct, yet which are not independent. Such a configuration would necessarily require the broadcasting of nonclassical information from the preparations. This follows directly from the fact that each preparation and transformation has only a single output (and input in the case of transformations) of each system type.

Let us then consider the case of two distinct observations, $A_i$ and $A_j$ with $i\neq j$. For a matter of simplicity, we consider $i=1$ and $j=2$, noting that every distinct pair $(i,j)$ is analogous. Since both observations are assumed to be tetrahedron-embeddable, each has three preparations in its causal past, concretely $A_1$ has $\neg (A_2), \neg (A_3), \neg (A_4)$ and $A_2$ has $\neg (A_1), \neg (A_3), \neg (A_4)$. Therefore, they may share at most two of them, i.e., their past coincide at most in $\neg (A_3)$ and $\neg (A_4)$. Hence, there are two possibilities for them to be non-independent: either they share both preparations or only one of them. In the former case, they are jointly tetrahedron-embeddable, as the closed circuit would be the one illustrated in Fig.~\ref{fig:S1}. However, the latter case cannot happen, as it would entail the cloning of nonclassical information. To see this, consider, without loss of generality, that $A_1$ and $A_2$ only share $\neg (A_3)$. In that case, the closed circuit must contain two copies of the preparation, $\neg (A_4)$. Then, there must be two transformations of the type $\neg (A_3A_4)$, so that the different copies of $\neg (A_4)$ are in the past of $A_1$ and $A_2$ respectively. Hence, as there is a unique preparation $\neg (A_3)$ in the past of both observations, one would need to broadcast the nonclassical information (concretely, the system type $4$)  so that both transformations of the type $\neg (A_3A_4)$ have the required inputs, proving that this case is not possible.

\end{proof}

Once the elementary case of single observations is established, we turn to the case of larger tetrahedron-embeddable sets. The same reasoning applies: if two tetrahedron-embeddable sets are non-independent, then their union must also be jointly tetrahedron-embeddable. %

\begin{proposition}
     Within $\mathcal{T}$, if two distinct non-independent sets of observations $S_1, S_2$ are tetrahedron-embeddable separately, then $S_1 \cup S_2$ is also tetrahedron-embeddable.
     \label{prop:closed_circuits}
\end{proposition}
\begin{proof}
    Note that an observation, $A_i$, can be jointly tetrahedron-embeddable with at most one observation of each different type, $A_j$ $(j\neq i)$ . The reason is that the causal past of $A_i$ contains three transformations that, respectively, can link it to at most one observation of each other type; otherwise, one would be forced to broadcast nonclassical information\footnote{Note that this argument can be made as well based on the impossibility of broadcasting the nonclassical information from the preparations.}. Together with the fact that two tetrahedron-embeddable observations of the same type must be independent (as explained in the proof of Proposition~\ref{singletons_proposition}), this implies that the union $S_1 \cup S_2$ contains at most one observation of each type.
    
    Also, note that the case where one of the sets is a subset of the other is trivial, hence let us consider the remaining possibilities. 

Now, as $S_1$ and $S_2$ are not independent, there must exist \emph{some} observation $s_1\in S_1$ that is non-independent with \emph{some} observation $s_2\in S_2\setminus S_1$. By Proposition~\ref{singletons_proposition}, the pair $(s_1,s_2)$ is therefore jointly tetrahedron-embeddable. However, for our purposes it is not enough to know that $s_2$ is jointly tetrahedron-embeddable with a single element $s_1$: we must show that $s_2$ is jointly tetrahedron-embeddable with the \emph{entire set} $S_1$. Establishing this intermediate step yields a recursive proof of Proposition~\ref{prop:closed_circuits}. Indeed, once $s_2$ is shown to be jointly tetrahedron-embeddable with all of $S_1$, we can enlarge the set by defining $S_1' = S_1 \cup \{s_2\}$ and repeat the same argument, iterating until $S_1$ grows to $S_1 \cup S_2$.

Hence, let us now prove such intermediate step, namely that $s_2$ is jointly tetrahedron-embeddable with $S_1$. This reduces to the following claim: if two distinct \emph{overlapping} sets of observations are tetrahedron-embeddable, namely the current set $S_1$ (which is updated at each iteration) and the pair $\{s_1,s_2\}$, then their union $S_1\cup\{s_1,s_2\}$ is also tetrahedron-embeddable.

Since $S_1\cup S_2$ contains at most four elements, and since the case $|S_1|=1$ is trivial (as $S_1$ would be a subset of $\{s_1,s_2\}$), only two nontrivial cases can arise in the recursive algorithm: either both overlapping sets are pairs, or one set is a triplet while the other is a pair. Importantly, in either of the two cases, proving that the union $S_1 \cup \{s_1,s_2\}$ is tetrahedron-embeddable requires showing that the tetrahedron-embeddable realizations of $S_1$ and of $\{s_1,s_2\}$ can be taken to share the \emph{same} underlying preparations. In other words, the preparations used to build the subcircuit influencing $\{s_1,s_2\}$ must coincide with those already used for $S_1$. We now treat the two cases separately:

\begin{adjustwidth}{0.1\columnwidth}{0pt} %
\noindent\textbf{Case 1: both sets are pairs.}
Assume both $S_1$ and $\{s_1,s_2\}$ are pairs and overlap in a single observation, say $s_1 = A_i$. Since $A_i$ appears in both circuits, its entire causal past is shared by both sets. In particular, the three preparations in the past of $A_i$ that are not $\neg(A_i)$ must coincide in the realizations of $S_1$ and of $\{s_1,s_2\}$.

\noindent Then, the only missing step is to show that the preparation $\neg(A_i)$ is also the same in both realizations. Let $t := S_1\setminus\{s_1\}$ denote the other element of the pair $S_1$. By construction, the causal past of $t$ contains three transformations, one of which links $t$ to an observation of the same type as $s_2$. This transformation must depend on one of the preparations already fixed above because $S_1$ is tetrahedron-embeddable. Equivalently, $s_2$ has such a transformation in its past which has to use the same fixed preparation (as $\{s_1,s_2\}$ is tetrahedron-embeddable). Then, as the fixed preparation cannot feed two different transformations of such type (one leading to $t$ and another leading to $s_2$) unless nonclassical information was broadcast, the relevant transformation must be the same in both realizations. This forces $\neg(A_i)$ to coincide as well. Hence, $S_1\cup\{s_1,s_2\}$ is tetrahedron-embeddable.

\noindent \textbf{Case 2: one set is a pair and the other a triplet.}
  Assume $S_1$ is a triplet and $\{s_1,s_2\}$ is a pair. By construction, the two sets overlap in exactly one observation, say $s_1=A_i$. As in Case~1, the presence of $A_i$ in both subcircuits forces the three preparations in its causal past to be identical for both realizations. That is, the three preparations distinct from $\neg (A_i)$ must be the same.

  \noindent It remains to show that both realizations use the \emph{same} $\neg (A_i)$ preparation. Let $S_1\setminus\{s_1\}=\{t,t'\}$ be the two other observations in the triplet. Applying the same argument as in the previous case to each of the pairs $\{s_1,t\}$ and $\{s_1,t'\}$, we conclude that both $t$ and $t'$ must share the preparation $\neg (A_i)$ with $s_2$. Therefore, all four preparations used to realize $S_1$ and $\{s_1,s_2\}$ are the same, and hence the union $S_1\cup\{s_1,s_2\}$ is tetrahedron-embeddable.
\end{adjustwidth}

\end{proof}

We now provide the probability rule that specifies the distributions generated by all possible closed circuits in the OPT $\mathcal{T}$.

\begin{definition}
    The probability rule of the OPT $\mathcal{T}$ is given by:
        \begin{itemize}
            \item Non-tetrahedron-embeddable observations: the first outcome occurs with probability 1 (and the second with probability 0). 
            \item Tetrahedron-embeddable sets of observations: the outcomes follow a shared random bit distribution, i.e., with probability $\tfrac{1}{2}$ all observations in the set yield the first outcome, and with probability $\tfrac{1}{2}$ all observations yield the second outcome. 
        \end{itemize}
\end{definition}

To see that the proposed probability rule is consistent, we have to check that it satisfies the principles of No-Signaling and Independence for all closed circuits. First, observe that even though observations that are not tetrahedron-embeddable may fail to be causally independent from other parties, the probability rule guarantees that their marginal distribution always factorizes from the rest of the parties: they output the first outcome with probability one regardless of their causal past. Then, the remaining part of any closed circuit consists of independent sets of tetrahedron-embeddable observations. For these, the principles of No-Signaling and Independence are ensured, since the two-layer tetrahedron causal structure imposes no factorization constraints on the probability distribution within each independent set. Moreover, it is straightforward to verify that the joint probability distribution factorizes across the partition into independent sets, in accordance with the prescribed rule.%

We remark that, as the probability rule yields the shared random bit distribution among the parties contained in each of the tetrahedron-embeddable sets of observations, the OPT $\mathcal{T}$ allows for perfect coordination in the causal structure of the two-layer tetrahedron (Fig.~\ref{fig:circuit-types}). Thus, this is a theory that allows for perfect coordination among four observations even though they do not share a common cause.

Finally, we show that the OPT $\mathcal{T}$ is causal in the sense of Ref.~\cite{d2017quantum}. This follows directly from the fact that the assigned probability rule is such that summing over any outcome yields the same marginal regardless of the observation associated to such outcome. This is equivalent to the fact that $\mathcal{T}$ possesses a unique deterministic effect for any system (recall that a deterministic effect is defined as the sum of all the events in an observation test).

\section{GHZ state requires a quantum common cause}
\label{fullyquantum_section}

In this section, we turn to the question of whether a genuinely quantum coordination task is achievable without a quantum common cause. To that end, note that we change the paradigm from having \emph{classical} observed nodes to \emph{quantum} observed nodes. This means that instead of producing a probability distribution, the parties (which now are described as quantum channels instead of measurements) produce a quantum state.

Previously, we defined perfect coordination as the task of outputting the probability distribution of a shared random bit. Then, the natural extension to the fully quantum case is to identify the idea of quantum perfect coordination as the task of producing a GHZ state,

\begin{equation}
    \ket{GHZ}_{A_1,...,A_N} = \frac{1}{\sqrt{2}}(\ket{0...0}+\ket{1...1}).
\end{equation}

Therefore, in this section, we provide a proof for the following theorem.

\begin{theorem}
    In quantum theory, obtaining a $\mathrm{GHZ}$ state requires a quantum common cause.
    \label{GHZtheorem}
\end{theorem}
We first give a proof for the perfect GHZ state and later a noise-robust version in terms of a Bell-like inequality. The noise-robust version is mostly motivated by the idea that the perfect GHZ state is practically impossible to achieve in experiments and hence an inequality is needed to certify the existence of a quantum common cause.

\subsection{Ideal GHZ state}

Note that the results of Sec.~\ref{sec:perfect_correlation} directly rule out the possibility of producing the GHZ states in the scenarios considered there (when we only change the observed nodes from classical to quantum), that is, scenarios in which the parties do not have a shared common cause. This can be understood because if the GHZ states were producible in such scenarios, the parties could simply measure their corresponding qubits in the Z basis and obtain the distribution of a shared random bit, thereby becoming perfectly coordinated. This would contradict our previous impossibility result. However, the set of DAGs that we study in this section are those in which all the parties do not share a \emph{quantum} common cause. In other words, a shared \emph{classical} common cause among all parties is allowed, which makes the problem nontrivial. Let us define this set formally:

\begin{definition}
    We define $\mathcal{G}_{Q,N}$ as the set of DAGs with $N$ quantum observed nodes and only classical and/or quantum latent nodes such that the quantum-observed nodes do not all share a quantum common cause.
    \label{DefinitionGQ4}
\end{definition}

Hence, following the same proof strategy we used in the case of classical perfect coordination, we first need to find a particular DAG which can generate every state achievable by some DAG in $\mathcal{G}_{Q,N}$. Naturally, this particular DAG, which we term $g^*_{Q,N}$ for each value of $N$, consists of the same causal structure described in Definition~\ref{def:gNstar} (considering now quantum observed nodes) but including a classical common cause, i.e., shared randomness between all the parties. Formally,

\begin{definition}\label{def:gQNstar}For a given integer $N$, we define the DAG $\textit{g}^*_{Q,N}$ to be $(V_Q,E_Q)$, where $V_Q$ is the set of all nodes in $\textit{g}^*_{Q,N}$ and $E_Q \subseteq V_Q \times V_Q$ the set of directed edges, each oriented from the first node in the pair to the second. Given $(V,E)$ as those defined in Definition~\ref{def:gNstar}:
\begin{itemize}
    \item $V_Q = \hat{V}\cup\{\lambda\}$,
    \item $E_Q = E\cup\left\{(\lambda,A_i)|\; i\in\{1,2,3\dots,N\} \right\}$,
\end{itemize}
where $\hat{V}$ is defined by taking the nodes in $V$ and imposing that they be quantum nodes and $\lambda$ is a classical latent node.
\end{definition}

The particular case of four parties, i.e., $g^*_{Q,4}$, is illustrated in Fig.~\ref{quantum-tetrahedron}. The proof of the fact that this is the most general DAG is completely analogous to the one given for Lemma~\ref{lemma:multilayerG4} in Appendix~\ref{app:proof1} so we just state the following lemma for completeness.

\begin{lemma}
    Any quantum state $\rho$ produced in a DAG $g \in \mathcal{G}_{Q,N}$ can also be produced in $g_{Q,N}^*$
    \label{quantum-contained}
\end{lemma}

\begin{figure*}
    \centering
    \begin{tikzpicture}[
  node distance=1cm and 1.5cm,
  every node/.style={draw, circle, minimum size=6mm, inner sep=0pt},
  level 1/.style={sibling distance=15mm},
  level 2/.style={sibling distance=10mm},
  arrow/.style={<-, >=latex}
]
  \definecolor{lightblue}{RGB}{173,216,230}
  \definecolor{lightpink}{RGB}{255,182,193}
  \definecolor{emerald}{RGB}{150, 240, 150}

  \node[draw, rectangle,  fill=lightpink] (A) {$A_1$};
  \node[draw,  rectangle, fill=lightpink, right=0.8cm of A] (B) {$A_2$};
  \node[draw,  rectangle, fill=lightpink, right= 0.8cm of B] (C) {$A_3$};
  \node[draw,  rectangle, fill=lightpink, right=0.8cm of C] (D) {$A_4$};

  \node[fill=lightpink, below left=1.2cm and 1.55cm of A] (AB) {$\neg (A_3A_4)$};
  \node[fill=lightpink, right=0.5cm of AB] (AC) {$\neg (A_2A_4)$};
  \node[fill=lightpink, right=0.5cm of AC] (AD) {$\neg (A_2A_3)$};
  \node[fill=lightpink, right=0.5cm of AD] (BC) {$\neg (A_1A_4)$};
  \node[fill=lightpink, right=0.5cm of BC] (BD) {$\neg (A_1A_3)$};
  \node[fill=lightpink, right=0.5cm of BD] (CD) {$\neg (A_1A_2)$};

  \node[fill=lightpink, below=3.4cm of A] (ABC) {$\neg (A_4)$};
  \node[fill=lightpink, below=3.4cm of B] (ABD) {$\neg (A_3)$};
  \node[fill=lightpink, below=3.4cm of C] (ACD) {$\neg (A_2)$};
  \node[fill=lightpink, below=3.4cm of D] (BCD) {$\neg (A_1)$};

 \node[draw, fill=emerald, above right=0.5cm and 0.18cm of B] (lambda) {$\lambda$};

  \foreach \s/\d in {A/AB, A/AC, A/AD, B/AB, B/BC, B/BD, C/AC, C/BC, C/CD, D/AD, D/BD, D/CD, A/lambda, B/lambda, C/lambda, D/lambda} {\draw[arrow] (\s) -- (\d);}

  \foreach \s/\d in {AB/ABC, AB/ABD, AC/ABC, AC/ACD, AD/ABD, AD/ACD, BC/ABC, BC/BCD, BD/ABD, BD/BCD, CD/ACD, CD/BCD} {\draw[arrow] (\s) -- (\d);}

\end{tikzpicture}
    \caption{The most general scenario in which four quantum observed nodes do not share a quantum common cause which we name $g^*_{Q,4}$. Note that sending the classical global common cause to the parties instead of the sources is completely general as we are not putting bounds on the dimensions of the latent nodes.%
    }
    \label{quantum-tetrahedron}
\end{figure*}
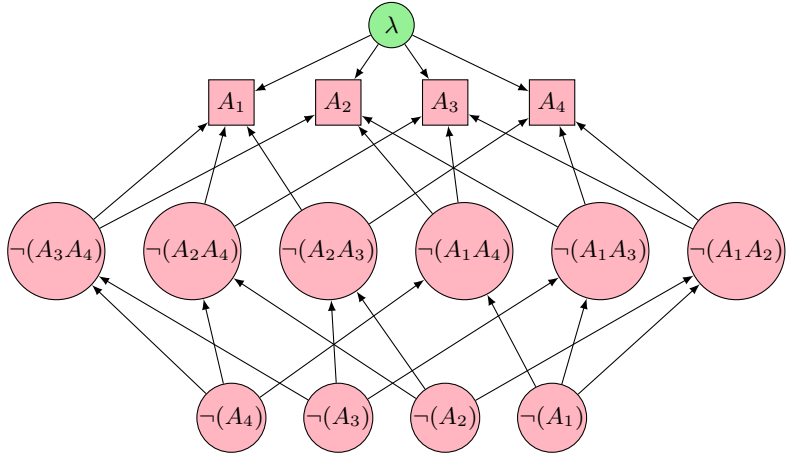

Equipped with the previous definition and lemma, we can prove the technical version of Theorem~\ref{GHZtheorem}.

\begin{theoremrevisited}{GHZtheorem}\label{analyticalghz}
Consider a DAG $g$ $\in \mathcal{G}_{Q,N}$. Then, the $N-$partite GHZ state is not achievable in $g$.
\end{theoremrevisited}

\begin{proof}
    By Lemma~\ref{quantum-contained}, we just need to show that the GHZ state is not achievable in the DAG $g_{Q,N}^*$. Now, recall that, as mentioned, the results of Sec.~\ref{sec:perfect_correlation} directly rule out the possibility of producing the GHZ states in the most general scenario (but with quantum observed nodes) considered there, i.e., $g_{N}^*$ (with quantum observed nodes). Then, by purity of the GHZ state, it directly follows that it is not producible in $g_{Q,N}^*$.

    To see this more concretely, note that $g_{Q,N}^*$ differs from $g_{N}^*$ (with quantum observed nodes) only by allowing additional shared randomness among all parties. Such shared randomness can at most produce convex mixtures of the states attainable without it. %
    Yet the GHZ state, being pure, is an extreme point of the convex set of states: any convex decomposition that yields GHZ must be trivial, i.e., supported entirely on GHZ states. Therefore, if for every value of the shared random variable the GHZ state is unattainable, then no mixture over that variable can produce a GHZ state.
    \end{proof}

\subsection{Inequality formulation}

Although the result for the perfect GHZ state is interesting from the theoretical and foundational point of view, the application of the result in experiments must come from a noise-robust version of it, as there are always some errors and finite statistics. Thus, we provide an inequality that must be satisfied by the correlations one would obtain by measuring the quantum observed nodes of any DAG in $\mathcal{G}_{Q,N}$, but that some measurements on the GHZ state would violate.

Before presenting the inequality, let us first explain the task from which the inequality is derived (that is the same as the one proposed in Ref.~\cite{Coiteux2021any}). The task consists in winning two games simultaneously. To play the two games, all parties perform two dichotomic measurements except for the second party, $A_2$, who performs three different measurements.
  
  The first game is the CHSH game between $A_1$ and $A_2$ (using the first two settings) conditioned on the product of the rest of the parties measuring the second setting. Then, defining the product as $\tilde{A}_{\text{rest}}^1 = A_3^1\cdot A_4^1 \cdots A_N^1$, we check a different instance of the CHSH game depending on the two possible values $1$ or $-1$. Concretely,

    \begin{equation}
\begin{split}
        I_{CHSH}^{\tilde{A}_{\text{rest}}^1=1} \circ \left\{ A_1 A_2\right\}= & \langle A^0_1 A^0_2 \rangle ^{\tilde{A}_{\text{rest}}^1=1} + \langle A^0_1 A^1_2 \rangle ^{\tilde{A}_{\text{rest}}^1=1} \\
       + & \langle A^1_1 A^0_2 \rangle ^{\tilde{A}_{\text{rest}}^1=1} - \langle A^1_1 A^1_2 \rangle ^{\tilde{A}_{\text{rest}}^1=1}
\end{split}
\end{equation}

   \begin{equation}
\begin{split}
        I_{CHSH}^{\tilde{A}_{\text{rest}}^1=-1} \circ \left\{ A_1 A_2\right\} = & \langle A^0_1 A^0_2 \rangle ^{\tilde{A}_{\text{rest}}^1=-1} + \langle A^0_1 A^1_2 \rangle ^{\tilde{A}_{\text{rest}}^1=-1} \\
       - & \langle A^1_1 A^0_2 \rangle ^{\tilde{A}_{\text{rest}}^1=-1} + \langle A^1_1 A^1_2 \rangle ^{\tilde{A}_{\text{rest}}^1=-1}
\end{split}
\end{equation}
  The second game is the ``same'' game among all the parties for setting 0 except for $A_2$ which uses input 2,
 \begin{equation}
        I_{\text{same}}= \langle A_1^0 A_2^2\rangle + \langle A_2^2 A_3^0\rangle  + \langle A_3^0 A_4^0 \rangle + \cdots + \langle A_{N-1}^0 A_N^0 \rangle.
    \end{equation}

Then, the noise-robust version of Theorem \ref{GHZtheorem} is

\begin{theoremnoise}{GHZtheorem}\label{last}
 Any correlation achieved by measuring a state produced in a DAG $g \in \mathcal{G}_{Q,N}$ according to some individual settings must satisfy

    \begin{equation}
     \begin{split}
\left(\frac{1-\langle \tilde{A}_{\operatorname{rest}}^1\rangle }{2}\right)^2 \left(I_{CHSH}^{\tilde{A}_{\operatorname{rest}}^1=-1}\circ \left\{ A_1A_2\right\}\right)^2 \\ + \left( \frac{1+\langle \tilde{A}_{\operatorname{rest}}^1 \rangle}{2}\right)^2 \left(I_{CHSH}^{\tilde{A}_{\operatorname{rest}}^1=1}\circ \left\{ A_1A_2\right\}\right)^2 \\ + 2(I_{\text{same}}\operatorname{cosec}(\frac{\pi}{2(N-1)})  - (N-1)\operatorname{cotan}(\frac{\pi}{2(N-1)}))^2 + 
\\ \leq \left[\left( \frac{1-\langle \tilde{A}_{\operatorname{rest}}^1 \rangle}{2}\right)^2 + \left( \frac{1+\langle \tilde{A}_{\operatorname{rest}}^1 \rangle}{2}\right)^2\right] 8.
\label{convexinequality}
\end{split}
\end{equation}
Specific measurements on the N-partite GHZ state violate the previous inequality: $LHS=4 + 2(N-1)^2(\operatorname{cosec}(\frac{\pi}{2(N-1)}) - \operatorname{cotan}(\frac{\pi}{2(N-1)}))^2 > 4=RHS$.
\end{theoremnoise}

\begin{proof}
By Lemma~\ref{quantum-contained}, showing that the previous inequality is valid for the scenario $g_{Q,N}^*$ (where each party measures the produced quantum state according to some settings) is a complete proof. Therefore, the particular DAG we consider in the case of $N=4$ is the one depicted in Fig.~\ref{GHZ_scenario}.

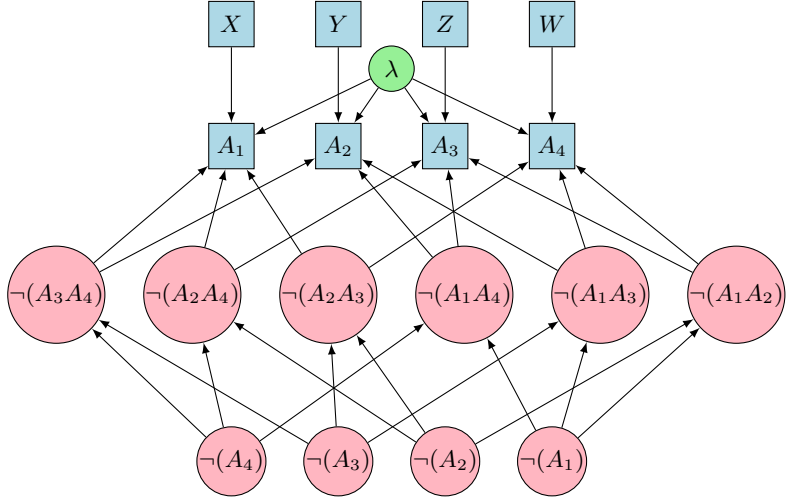
\begin{figure*}
    \centering
    \begin{tikzpicture}[
  node distance=1cm and 1.5cm,
  every node/.style={draw, circle, minimum size=6mm, inner sep=0pt},
  level 1/.style={sibling distance=15mm},
  level 2/.style={sibling distance=10mm},
  arrow/.style={<-, >=latex}
]
  \definecolor{lightblue}{RGB}{173,216,230}
  \definecolor{lightpink}{RGB}{255,182,193}
  \definecolor{emerald}{RGB}{150, 240, 150}

  \node[draw, rectangle,  fill=lightblue] (A) {$A_1$};
  \node[draw,  rectangle, fill=lightblue, right=0.8cm of A] (B) {$A_2$};
  \node[draw,  rectangle, fill=lightblue, right= 0.8cm of B] (C) {$A_3$};
  \node[draw,  rectangle, fill=lightblue, right=0.8cm of C] (D) {$A_4$};
  \node[draw, rectangle,  fill=lightblue, above=1cm of A] (X) {$X$};
  \node[draw,  rectangle, fill=lightblue, right=0.8cm of X] (Y) {$Y$};
  \node[draw,  rectangle, fill=lightblue, right= 0.8cm of Y] (Z) {$Z$};
  \node[draw,  rectangle, fill=lightblue, right=0.8cm of Z] (W) {$W$};

  \node[fill=lightpink, below left=1.2cm and 1.55cm of A] (AB) {$\neg (A_3A_4)$};
  \node[fill=lightpink, right=0.5cm of AB] (AC) {$\neg (A_2A_4)$};
  \node[fill=lightpink, right=0.5cm of AC] (AD) {$\neg (A_2A_3)$};
  \node[fill=lightpink, right=0.5cm of AD] (BC) {$\neg (A_1A_4)$};
  \node[fill=lightpink, right=0.5cm of BC] (BD) {$\neg (A_1A_3)$};
  \node[fill=lightpink, right=0.5cm of BD] (CD) {$\neg (A_1A_2)$};

  \node[fill=lightpink, below=3.4cm of A] (ABC) {$\neg (A_4)$};
  \node[fill=lightpink, below=3.4cm of B] (ABD) {$\neg (A_3)$};
  \node[fill=lightpink, below=3.4cm of C] (ACD) {$\neg (A_2)$};
  \node[fill=lightpink, below=3.4cm of D] (BCD) {$\neg (A_1)$};

 \node[draw, fill=emerald, above right=0.5cm and 0.18cm of B] (lambda) {$\lambda$};

\draw[arrow] (A) -- (X);
\draw[arrow] (B) -- (Y);
\draw[arrow] (C) -- (Z);
\draw[arrow] (D) -- (W);
  \foreach \s/\d in {A/AB, A/AC, A/AD, B/AB, B/BC, B/BD, C/AC, C/BC, C/CD, D/AD, D/BD, D/CD, A/lambda, B/lambda, C/lambda, D/lambda} {\draw[arrow] (\s) -- (\d);}

  \foreach \s/\d in {AB/ABC, AB/ABD, AC/ABC, AC/ACD, AD/ABD, AD/ACD, BC/ABC, BC/BCD, BD/ABD, BD/BCD, CD/ACD, CD/BCD} {\draw[arrow] (\s) -- (\d);}

\end{tikzpicture}
    \caption{Scenario arising from measuring the observed nodes in $g_{Q,4}^*$ according to some settings.}
    \label{GHZ_scenario}
\end{figure*}

The proof is based on the quantum inflation method, so we start by providing the particular quantum inflation we construct. As explained in Sec.~\ref{sec:perfect_correlation}, we can specify it by just providing the assignment of which
copies of the sources are in the causal past of each observed node (that is, specify the set of copy indices of
each observed node). Note that, in this case, we are considering a scenario with settings in each party and hence we need to specify the set of copy indices for each setting of each observed node. The quantum inflation is specified by Table~\ref{inflation_indices}.

Note that, in the inflation, the classical common cause is distributed to all the parties of the inflation as classical information can be broadcasted. Importantly, the constructed inflation includes the original structure for the operators involved in the conditioned CHSH game (that is, the operators $A_1^0, A_1^1,A_2^0,A_2^1,A_3^1,A_4^1,...,A_N^1$ share the same causal past as in the original scenario) as well as the \emph{quantum cut inflation}
for all the operators involved in the ``same'' game (that is, the part of the inflation involving the operators $A_1^0,A_2^2,A_3^0,A_4^0,...,A_N^0$ is the same as the one used in Sec.~\ref{sec:perfect_correlation} which can be graphically represented for the case of four parties as in Fig.~\ref{fig:inflation-tetrahedron}). This fact means that the operators involved in the conditioned CHSH game form an injectable set %
and that all adjacent pairs of operators involved in the ``same'' game also form injectable sets.

\begin{table}[h]
    \center
\begin{tabular}{|c|c|c|c|c|c|c|}
\hline
                     & $\neg (A_1)$ & $\neg (A_2)$ & $\neg (A_3)$ & ... & $\neg A_N$ & $\lambda$ \\ \hline
$A_1^{0,1}$ & - & 1 & 1 & 1   & 1 & 1 \\ \hline
$A_2^{0,1}$ & 1 & - & 1 & 1   & 1 & 1 \\ \hline
$A_3^{1}$ & 1 & 1 & - & 1   & 1 & 1 \\ \hline

...                  & 1 & 1 &  1 & -   & 1 & 1 \\ \hline
$A_N^{1}$ & 1 & 1 & 1 & 1   & - & 1 \\ \hline
$A_2^{2}$ & 2 & - & 1 & 1   & 1 & 1 \\ \hline
$A_3^{0}$ & 2 & 2 & - & 1   & 1 & 1 \\ \hline
...                  & 2 & 2 &  2 & -   & 1 & 1 \\ \hline
$A_N^{0}$ & 2 & 2 & 2 & 2   & - & 1 \\ \hline
\end{tabular}
\caption{Copy indices of the sources (columns) in the causal past of every observed node receiving a specific setting (rows). The table describes the quantum inflation used to prove Theorem~\nameref{last}.}
\label{inflation_indices}
\end{table}

Now, we follow similar steps as in Appendix D of Ref.~\cite{cao2022experimental}. We shall start from the quantum monogamy inequality between the violation of the CHSH inequality and perfect coordination of one of its players with a different party introduced in Ref.~\cite{augusiak2014elemental}. %
In our case, we condition on $\tilde{A}_{\text{rest}}^1=1$ or $-1$,

\begin{equation}
\left(I_{CHSH}^{\tilde{A}_{\text{rest}}^1=1}\circ \left\{ A_1A_2\right\}\right)^2 + 4\langle A_1^0 A^0_N \rangle_{\tilde{A}_{\text{rest}}^1=1}^2 \leq 8.
\end{equation}
\begin{equation}
\left(I_{CHSH}^{\tilde{A}_{\text{rest}}^1=-1}\circ \left\{ A_1A_2\right\}\right)^2 + 4\langle A_1^0 A_N^0 \rangle_{\tilde{A}_{\text{rest}}^1=-1}^2 \leq 8.
\end{equation}
 Then, doing a positive linear combination of the previous inequalities weighted by the squares of the probabilities of $\tilde{A}_{\text{rest}}^1=1$ or $-1$, respectively, we obtain

\begin{equation}
 \begin{split}
\left(\frac{1-\langle \tilde{A}_{\text{rest}}^1\rangle }{2}\right)^2 \left[\left(I_{CHSH}^{\tilde{A}_{\text{rest}}^1=-1}\circ \left\{ A_1A_2\right\}\right)^2 + 4\langle A_1^0 A_N^0 \rangle_{\tilde{A}_{\text{rest}}^1=-1}^2 \right] \\ + \left( \frac{1+\langle \tilde{A}_{\text{rest}}^1 \rangle}{2}\right)^2 \left[\left(I_{CHSH}^{\tilde{A}_{\text{rest}}^1=1}\circ \left\{ A_1A_2\right\}\right)^2 + 4\langle A_1^0 A_N^0 \rangle_{\tilde{A}_{\text{rest}}^1=1}^2 \right]  \\ \leq \left[\left( \frac{1-\langle \tilde{A}_{\text{rest}}^1 \rangle}{2}\right)^2 + \left( \frac{1+\langle \tilde{A}_{\text{rest}}^1 \rangle}{2}\right)^2\right] 8.
\label{positive_combination_inequality}
\end{split}
\end{equation}

Manipulating the previous expression, one can write it as

\begin{equation}
     \begin{split}
\left(\frac{1-\langle \tilde{A}_{\text{rest}}^1\rangle }{2}\right)^2 \left(I_{CHSH}^{\tilde{A}_{\text{rest}}^1=-1}\circ \left\{ A_1A_2\right\}\right)^2 \\ + \left( \frac{1+\langle \tilde{A}_{\text{rest}}^1 \rangle}{2}\right)^2 \left(I_{CHSH}^{\tilde{A}_{\text{rest}}^1=1}\circ \left\{ A_1A_2\right\}\right)^2 \\ + 2\langle A_1^0 A_N^0 \rangle^2 + 2\langle A_1^0 A_N^0 \tilde{A}_{\text{rest}}^1 \rangle^2
\\ \leq \left[\left( \frac{1-\langle \tilde{A}_{\text{rest}}^1 \rangle}{2}\right)^2 + \left( \frac{1+\langle \tilde{A}_{\text{rest}}^1 \rangle}{2}\right)^2\right] 8.
\label{positive_combination_inequality_2}
\end{split}
\end{equation}
where $\langle A_1^0 A_N^0 \tilde{A}_{\text{rest}}^1 \rangle$ is the tripartite correlator.

Now, we use lower bounds for the terms $\langle A_1^0 A_N^0 \rangle^2$ and $\langle A_1^0 A_N^0 \tilde{A}_{\text{rest}}^1 \rangle^2$. Concretely, for the latter, we use positivity of a square and, for the former, we use the SDP-based inequality (explained in Appendix \ref{SDP_appendix}) coming from the part of the inflation that corresponds to the \emph{quantum cut inflation},

\begin{equation}
\label{NewIneq}
\begin{split}
\langle A_1^0 A_N^0 \rangle \geq \\ I_{\text{same}}\text{cosec}\left(\frac{\pi}{2(N-1)}\right)
- (N-1)\text{cotan}\left(\frac{\pi}{2(N-1)}\right).
\end{split}
\end{equation}
Note, however, that the previous inequality is a bound for $\langle A_1^0 A_N^0\rangle$. Therefore, to have a lower bound on its square, we can just square the previous inequality under the assumption that both sides are positive. That means that we assume that we are in the range of values of $I_{same}$ where the right-hand side of the previous inequality is positive. %
Therefore, substituting $\langle A_1^0 A_N^0 \rangle^2$ in Eq.~\eqref{positive_combination_inequality_2} and using this bound, we obtain the inequality of Eq.~\eqref{convexinequality}.

Remarkably, all the probabilities involved in the inequality of Eq.~\eqref{convexinequality} correspond to probabilities of injectable sets, and therefore, the inequality applies to the original scenario.

Finally, let us recall the protocol proposed in \cite{Coiteux2021any} to achieve the quantum maximum in both games simultaneously. The strategy consists of all players measuring in the Z basis for the inputs corresponding to the ``same'' game and, for the inputs of the \textit{CHSH} game, the last $N-2$ players measure in the X basis while the first two players use optimal measurements for the standard Bell game, centered on $A_1$ measuring in the Z basis on input $X=0$. Then, the resulting value for the left-hand side of Eq.~\eqref{convexinequality} is  $4 +2(N-1)^2(\text{cosec}(\frac{\pi}{2(N-1)})-\text{cotan}(\frac{\pi}{2(N-1)}))^2$ while the right-hand side is $4$, thus violating the inequality.
\end{proof}

\subsection{Noise tolerance analysis}
\label{noise_section}

Finally, we provide a noise tolerance study of the inequality presented to certify the existence of a quantum common cause (Eq.~\eqref{convexinequality}).
In our analysis, we consider the white noise model for the GHZ state:
\begin{equation}
    \rho_{v} = v|GHZ_N\rangle\langle GHZ_N|+(1-v)\frac{I_N}{2^N},
    \label{noise_model_GHZ}
\end{equation}
where $v\in[0,1]$ is the noise parameter. In particular, $v=1$ corresponds to the perfect GHZ state while $v=0$ to the maximally mixed state.

The fidelity of $\rho_{v}$ with respect to the GHZ is given by $f\coloneqq\langle GHZ_N|\rho_{v}|GHZ_N\rangle$, therefore,
\begin{equation}
    f = \frac{1+v(2^N-1)}{2^N}.
\end{equation}

Now, if we apply to the noisy state the measurements that yield the maximal violation of Eq.~\eqref{convexinequality} for a perfect GHZ state, we obtain the values $I_{CHSH}^{\tilde{A}^1_{\text{rest}}=\pm1}\circ \left\{ A_1A_2\right\}[\rho_v]=v2\sqrt{2}$, $I_{\text{same}}[\rho_v]=v(N-1)$ and $\langle \tilde{A}_{\text{rest}}^1\rangle=0$. Thus, the inequality of Eq.~\eqref{convexinequality} is violated if

\begin{equation}\label{noisetol_ineq}
\begin{split}
v^2 + \frac12\Bigl(v(N-1)
\text{cosec}\!\left(\frac{\pi}{2(N-1)}\right)
\\-(N-1)\text{cotan}\!\left(\frac{\pi}{2(N-1)}\right)
\Bigr)^2 -1 > 0.
\end{split}
\end{equation}
That is, for a given number of parties $N$, finding the violation reduces to studying the positivity of a second order polynomial.

It is important to recall that in our derivation of Eq.~\eqref{convexinequality}, we assumed the positivity of the lower bound of $\langle A_1^0 A_N^0 \rangle$, i.e., positivity of the right-hand side of the inequality of Eq.~\eqref{NewIneq}. Hence, our noise analysis must restrict to the values of the parameter noise which satisfy the following restriction:

\begin{equation}
    v \geq \text{cos} \left(\frac{\pi}{2(N-1)}\right).
\end{equation}

In the region of the noise parameter where the previous restriction is satisfied, there is only one root of the second order polynomial of Eq.~\eqref{noisetol_ineq} which we call $v_{\text{min},N}$. Then, as the leading coefficient of the polynomial is positive, $ v_{\text{min},N}$ is the limiting value of $v$ to observe a violation of inequality~\eqref{convexinequality}, i.e., any value $v>v_{\text{min},N}$ ensures the presence of a quantum common cause. The corresponding value of the fidelity to the noise parameter $v_{\text{min},N}$ is denoted as $f_{\text{min},N}$. The inequality violations for $N=4,5,\dots,10$ are depicted in Fig.~\ref{fig:IneqViolation}. Numerical values (up to the fourth decimal) for $v_{\text{min},N}$ and $f_{\text{min},N}$ are provided in table \ref{tab:p_f_min}.

\begin{table}[h]
    \centering
\begin{tabular}{|c|c|c|}
\hline
$N$ & $v_{\text{min},N}$ & $f_{\text{min},N}$ \\ \hline
4  & 0.9439 & 0.9474 \\ \hline
5  & 0.9612 & 0.9624 \\ \hline
6  & 0.9717 & 0.9721 \\ \hline
7  & 0.9785 & 0.9787 \\ \hline
8  & 0.9831 & 0.9832 \\ \hline
9  & 0.9864 & 0.9865 \\ \hline
10 & 0.9889 & 0.9889 \\ \hline
\end{tabular}
\caption{Numerical values (up to the fourth decimal) of $p_{\min}$ and $f_{\min}$ for different values of $N$.}
\label{tab:p_f_min}
\end{table}

\begin{figure*}[]
    \centering
    \includegraphics[width=10cm]{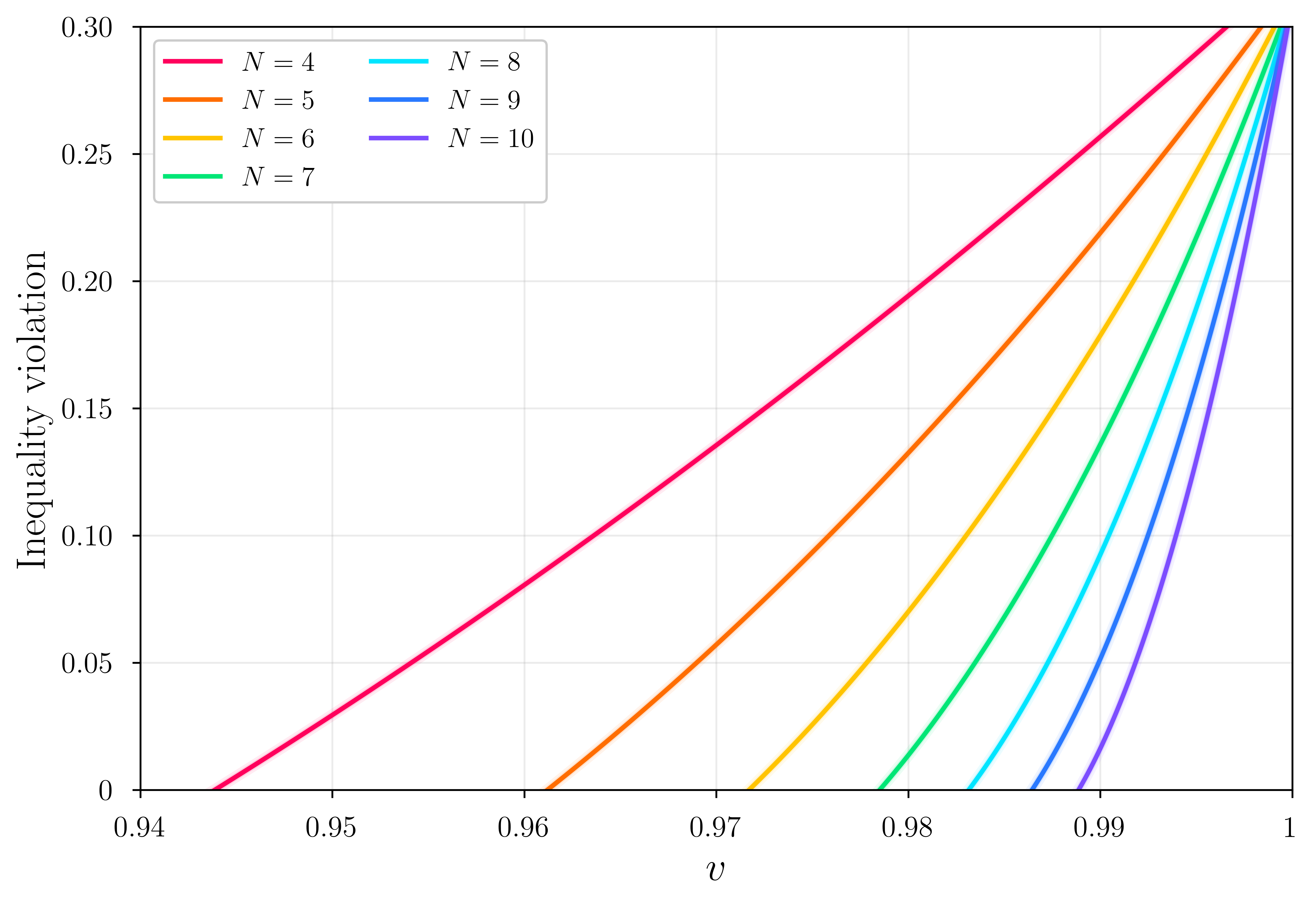}
    \caption{Value of the inequality violation, i.e., value of the left-hand side of Eq.~\eqref{noisetol_ineq}, as a function of the noise parameter $v$ for different values of $N$. Note that the x-intercept correspond to to $v_{\text{min},N}$.}
    \label{fig:IneqViolation}
\end{figure*}

In appendix~\ref{appendix_alternative}, we provide an inequality that slightly improves the noise tolerance only for the case of $N=4$, concretely $v_{\text{min},4}=0.9417$ and $f_{\text{min},4}=0.9454$.

For $N\geq 5$, the required threshold is, to the best of our knowledge, beyond current experimental capabilities. In contrast, the case of $N=4$ is compatible with the fidelities reported in recent multipartite GHZ experiments~\cite{cao2022experimental}. However, existing four-photon GHZ implementations rely on strong post-selection, retaining only those experimental runs in which all photons are successfully detected. As a result, the effective selection of valid events is not guaranteed to be independent of the measurement process, and thus such experiments do not provide a certification of a common cause.

A conclusive test of our inequality would instead require an event-ready, heralded GHZ source, in which the heralding signal is independent of all measurement settings and outcomes. While heralded GHZ generation has been demonstrated for three parties~\cite{cao2024photonic,chen2024heralded}, extending these schemes to four parties remains an outstanding experimental challenge.

\section{Discussion}
\label{discussion}

In this work, we addressed the question of whether there exist causal principles, beyond No-Signaling and Independence, that further constrain the space of physically reasonable theories
compatible with the correlations observed in nature. To this end, we introduced the Coordination principle, which can be viewed as a genuinely multipartite extension of the Reichenbach principle (or equivalently, of its contrapositive, the Independence principle). It states that $N$ variables can be perfectly coordinated only if they share a common cause. We showed that this principle is not implied by No-Signaling and Independence alone by constructing an explicit theory of information (formulated within the framework of operational probabilistic theories) that satisfies them while nevertheless enabling perfect coordination among parties without a common cause. Crucially, this separation relies on considering the most general causal scenarios, in particular allowing intermediate transformations, which are essential for the post-quantum counterexample.

In contrast, we proved that quantum theory satisfies the Coordination principle. We did so by analyzing the most general causal scenarios in which the observed parties do not all share a common cause. Using the quantum inflation method, we derived Bell-like inequalities that certify the presence of a shared common cause in a noise-tolerant way. These inequalities, however, should be viewed as a first step toward the strongest operational constraints implied by the Coordination principle in quantum theory: our bounds were derived from a specific inflation together with a first-level NPA relaxation. Certainly, tighter inequalities can be found by considering larger inflations or higher levels of the NPA hierarchy but scalability is a major bottleneck: pushing to higher levels rapidly becomes extremely expensive computationally.

Lastly, we turned to a genuinely quantum coordination task: the preparation of a $N-$partite GHZ state. We showed that even granting the parties arbitrary shared \emph{classical} randomness, perfect quantum coordination still requires a \emph{quantum} common cause. In other words, a $N-$partite GHZ state cannot be generated in the most general scenarios in which the parties do not all share a quantum common cause. Motivated by experimental realities, where noise is unavoidable, we also derived Bell-like inequalities that witness the presence of a quantum common cause in a noise-robust way.

An interesting direction for future work is to clarify the relation between our results and Ref.~\cite{coiteux2025genuinely}, which introduces a notion of genuinely multipartite entanglement enabled by intermediate transformations. That work focuses on a specific network with transformations (used to represent communication or interaction among close systems) rather than the fully general causal scenarios considered here. It would be valuable to formulate an analogous classification of multipartite entanglement relative to the most general scenarios, and to compare it with the more restricted class studied in Ref.~\cite{coiteux2025genuinely}.

More importantly, our results motivate the search for a \emph{quantitative} version of the Coordination principle. That is, a principle that not only avoids the possibility of \emph{perfect} correlation when there is no common cause but that also quantitatively constrains how well $N$ parties which do not share a common cause can coordinate in any reasonable causal information theory. However, this task seems to be challenging as developing such bounds will require new technical tools beyond those currently established (namely, non-fanout inflation).

Finally, let us recall that throughout this work we instantiated perfect coordination as the task of producing a shared random bit. While this is arguably the canonical coordination task, other natural notions of multipartite coordination can be considered.
We speculate that the quantitative version of the Coordination principle should constrain such alternative coordination tasks as well. Nevertheless, the operational probabilistic theory we provided here can be straightforwardly modified to accommodate other coordination targets.

\begin{acknowledgments}
We thank Victor Gitton, Tein van der Lugt, Marina Maciel Ansanelli, Roberto D. Baldij\~{a}o, Fatemeh Moradi, Xiangling Xu, Peter Brown, Augustin Vanrietvelde, and Renato Renner for fruitful discussions. AC would like to thank Prof. Marc-Olivier Renou and Prof. Renato Renner for supervising his master's thesis, which in part led to the present publication. 
A.C., L.T. and M.-O. R. acknowledge funding by the ANR for the JCJC grants LINKS (No. ANR-23-CE47-0003), the T-ERC QNET (No. ANR-24-ERCS-0008), the project QUANTINT, as well as the European Union’s Horizon 2020 Research and Innovation Programme under QuantERA Grant Agreements No. 731473 and No. 101017733. This work was funded by the European Union under the Marie Skłodowska-Curie Actions (MSCA) through the QNETS project (grant agreement ID: 101208259). Views and opinions expressed are, however, those of the author(s) only and do not necessarily reflect those of the European Union or the European Education and Culture Executive Agency (EACEA). Neither the European Union nor EACEA can be held responsible for them.
MCA and DC also acknowledge support from the Natural Sciences and Engineering Research Council of Canada (grants 50505-11449 and 50505-11450). Research at Perimeter Institute is supported in part by the Government of Canada through the Department of Innovation, Science and Economic Development and by the Province of Ontario through the Ministry of Colleges and Universities.

\end{acknowledgments}

\bigskip
\nocite{apsrev42Control}
\setlength{\bibsep}{1pt plus 1pt minus 1pt}
\bibliographystyle{apsrev4-2-wolfe}
\nocite{MasterThesisAntoine}
\bibliography{Refs}

\newpage

\appendix
\section{Proof of Lemma \ref{lemma:multilayerG4}}
\label{app:proof1}

In this appendix we formally prove Lemma~\ref{lemma:multilayerG4}. However, let us first show a proposition that serves as an intermediate step in order to demonstrate Lemma~\ref{lemma:multilayerG4}. Note that the following proposition is written for the general case of N parties but, to show Lemma~\ref{lemma:multilayerG4}, we just need the case of N=4.

\begin{proposition}
    Any DAG \textit{g} $\in \mathcal{G}_N$ \footnote{$\mathcal{G}_N$ is defined in Definition~\ref{def:setN}.}, in which not all the observed nodes are terminal\footnote{A \emph{terminal} node is a node with no children, i.e., it is not the cause of any other. Consequently, a DAG in which not all observed nodes are terminal necessarily contains at least one directed edge between observed nodes, representing a direct causal influence among them.} is observationally contained in another DAG \textit{g'} $\in \mathcal{G}_N$ where all the observed nodes are terminal.
    \label{terminalN}
\end{proposition}
\begin{proof}
The proof consists in following some steps to construct $g'$ starting from $g$. Then, we justify why each of these steps does not constrain the set of achievable probability distributions. Therefore, showing that $g'$ observationally contains $g$. The steps are the following:
\begin{enumerate}
    \item Take all the largest sets of observed nodes that share a common cause and, for each of them, add a nonclassical latent node which has all the nodes in the set as children.
    \item Find the largest sets of common descendants for all sets of the \emph{added} nonclassical common causes (note that as we are considering common descendants, the sets of common causes has size two or more). For each distinct maximal set of common descendants of size larger or equal than two, $\mathcal{S}_i$ (where $i$ is an index for each of these sets), take the set of all their common causes among the \emph{added} ones in the first step, $\mathcal{C}_{\mathcal{S}_i}$. Then, add an intermediate latent node which has as parents all the common causes in $\mathcal{C}_{\mathcal{S}_i}$ and, as children, all the observed nodes in $\mathcal{S}_i$. Then, we shall repeat this step considering only the recently added generation of intermediate latent nodes until the maximal sets of common descendants of any pair of intermediate nodes are only of size 1. 
    \item Remove direct arrows between observed nodes.
\end{enumerate}

The first and second steps only add components to the causal scenario and, therefore, the resulting DAGs ($g_1$ and $g_2$, respectively) observationally contain the initial DAG $g$. Hence, the missing task is to prove that the DAG after the third step, $g'$, observationally contains $g_2$.

In order to show it, note first that if two observed nodes are directly connected, say $A\rightarrow B$, then every ancestor of $A$ is also an ancestor of $B$, i.e., $Anc(A)\subseteq Anc(B)$, where $Anc(X)$ means the set of ancestors of $X$. By construction, this entails the following structural feature for any directed edge between observed nodes in $g_2$: There is a unique latent node, $\mathcal{L}$, that is a parent of $A$ and $B$ and whose causal past contains all other ancestors of $A$.
To see it more clearly, if $A$ had two distinct such latent parents, then both of them would also lie in the causal past of $B$. In that case, these two latent nodes would share at least two common descendants, namely $A$ and $B$, and the construction would trigger step 2 again; hence such a situation is excluded.

Once this is established, we can explain why removing arrows between observed nodes does not change the set of achievable probability distributions (i.e. why $g'$ contains $g_2)$ by virtue of the following protocol. %
The measurement that $A$ would have applied locally (to the system $A$ would have received) is performed in $\mathcal{L}$ who then sends the resulting outcome to both $A$ and $B$ (who also receives its corresponding nonclassical system). Then $A$ simply outputs this value, while $B$ uses it exactly as it would have used the message received from $A$ along the edge $A\rightarrow B$. This reproduces the original behavior while eliminating the need of the direct causal arrow between the observed nodes.

The same idea applies to chains of observed nodes, e.g., $A\rightarrow B \rightarrow C$. One applies the above replacement edge by edge, in causal order. For each edge except the last, the corresponding unique latent parent forwards the simulated outcome not only to the immediate observed nodes, but also to a latent node that lies in the causal past of the remaining nodes in the chain, ensuring that downstream dependencies are preserved.

\end{proof}

We shall now turn to proving Lemma~\ref{lemma:multilayerG4}. We state it again here for the convenience of the reader.

\begin{replemma}{lemma:multilayerG4}
Any DAG $\textit{g} \in \mathcal{G}_4$ is observationally contained in the \hyperref[2-layer-tetrahedron]{two-layer tetrahedron}.
\end{replemma}

\begin{proof}
We first invoke Proposition~\ref{terminalN} for the case of $N=4$ to say that, without loss of generality, we can consider only DAGs $\textit{g} \in \mathcal{G}_4$ which have only terminal observed nodes. Then, for any such DAG, we give a general procedure to extend it to the \hyperref[2-layer-tetrahedron]{two-layer tetrahedron}. Therefore, as we are just adding more components to the causal scenario, we conclude that any DAG $\textit{g} \in \mathcal{G}_4$ is observationally contained in the \hyperref[2-layer-tetrahedron]{two-layer tetrahedron}.

The mentioned extension is the following. First, we arrange $\textit{g}$ into a special form which corresponds to the intuitive notion of a genealogical tree, in which we arrange our nodes into generations according to the following rules:

\begin{itemize}
\item We label the observed nodes by a number: $1,2,3,4$.

\item The latent nodes are labeled by the set of its observed descendants, $S$. If $S$ is the same for $k \in \mathbb{N}$ nodes, we distinguish them by labeling them by a pair $(S,i)$ where $i=1,2,...k$.

\item We order the DAG by generations with generation $j \in {1,2,3}$ containing all nodes with sets of cardinality $j$ on their label. (Notice that, since there is no shared cause between all observed parties, there can be at most $3$ generations).
\end{itemize}

The above procedure is always possible given that $\textit{g}$ is a DAG. It is noteworthy that, by virtue of the chosen labeling of the nodes, a causal edge propagating within the same generation can only connect nodes that share the same set of descendants (otherwise, the influencing node would inherit additional descendants and would belong to a different generation).

Once \textit{g} has been genealogically rearranged, we are free to add components (nodes and/or edges) to the causal scenario because doing so does not restrict the set of probability distributions compatible with the DAG. We build $g'$ such that it extends $g$ as follows:
\begin{itemize}
\item We merge all nodes with the same set $S$ in their label, denoting by $S$ the node obtained from the merger. 

\item We allow all present latent nodes to be quantum latent nodes.
\item We add quantum latent nodes to the graph such that all elements in $\mathcal{P}\left(\{1,2,3,4\}\right)\backslash \left\{\emptyset, \{1,2,3,4\}\right\}$ are represented, where $\mathcal{P}\left(\{1,2,3,4\}\right)$ denotes the power set of the set $\{1,2,3,4\}$.
\item We add arrows from each node $S$ to all the nodes $S'$ if $S'$ is a subset of $S$ with cardinality $|S|-1$.
\end{itemize}

Finally, notice that $g'$ is exactly the causal structure of the \hyperref[2-layer-tetrahedron]{two-layer tetrahedron}. Therefore, $g$ is observationally contained in the \hyperref[2-layer-tetrahedron]{two-layer tetrahedron}.
\end{proof}

Remarkably, the previous proof can be straightforwardly generalized to the case of $N$ parties which constitute the proof of Lemma~\ref{lemma:multilayerG}. Additionally, it does not rely on the classicality of observed nodes but rather on the topology of the causal structure of the DAGs. Therefore, it is also a valid proof for the case of having observed quantum nodes, i.e., showing Lemma~\ref{quantum-contained}.

\section{SDP inequalities}
\label{SDP_appendix}

In this appendix, we present the formal proofs of Theorems~\ref{noiserobusttheorem4} and~\nameref{noiserobusttheorem} as well as the inequality used in the proof of Theorem~\nameref{last}. The technical core of the argument consists in deriving inequalities as certificates from a semidefinite program (SDP). 

We begin by explaining how to obtain an inequality involving only two-party correlators of adjacent parties, together with the correlator between the first and last parties because such inequalities are the ones used in the proof of Theorem~\nameref{last}. We then specialize these general constructions to derive the certificates of Theorems~\ref{noiserobusttheorem4} and~\nameref{noiserobusttheorem}.

Let us first provide the details for the four-partite case and then generalize it for $N$ parties. We consider a moment matrix where the only expressible sets are the two-party correlators of adjacent pairs plus the pair of first and last parties, that is, $(A_i,A_{i+1}) \forall i \in \{1,2,3\}$ and $(A_1,A_4)$. (Note that this is the case for the \emph{quantum cut inflation} defined by Table~\ref{tab:quantum_cut4}). Concretely, we use the first level of the NPA hierarchy, that is, we construct a moment matrix including only sequences of operators of length one. In particular,
\begin{equation*}
\begin{aligned}
& \text{\hspace{10pt}$A_1$ \hspace{10pt}$A_2$ \hspace{19pt}$A_3$ \hspace{22pt}$A_4$} \\
 \Gamma=&\left(\begin{array}{cccc}
1 & \langle A_1 A_2\rangle & x_1 & \langle A_1 A_4\rangle \\
& 1 & \langle A_2 A_3\rangle & x_2 \\
& & 1 & \langle A_3 A_4\rangle \\
& & & 1
\end{array}\right),
\end{aligned}    
\end{equation*}
where we only give the upper triangular part of $\Gamma$ since it is symmetric. We write the correlators for all expressible sets and leave as variables, $x_i$, the elements for which we lack information. As $\Gamma$ is positive semidifinite (PSD), the matrix resulting from multiplying $\Gamma$ with any other PSD matrix will be also PSD; therefore, it must have positive trace.  Then, in order to obtain an inequality, we need to find a positive semidefinite matrix, i.e., the witness $W$, such that the trace of the multiplication only involves the correlators of expressible sets. The witness can be found using the duality of the semidefinite program for the positive completion problem of $\Gamma$ (notice that the witness is not unique). In particular, we find

\begingroup
\renewcommand{\arraystretch}{1.5}
\begin{equation*}
 W_4=\left(\begin{array}{cccc}
\frac{1}{2}\cos(\frac{\pi}{6}) & -\frac{1}{2} & \phantom{-}0 & \phantom{-}\frac{1}{2}\sin(\frac{\pi}{6}) \\
-\frac{1}{2} & \phantom{-}\cos(\frac{\pi}{6}) & -\frac{1}{2} & \phantom{-}0 \\
\phantom{-}0 & -\frac{1}{2} & \phantom{-}\cos(\frac{\pi}{6}) & -\frac{1}{2} \\
\phantom{-}\frac{1}{2}\sin(\frac{\pi}{6}) & \phantom{-}0 & -\frac{1}{2} & \frac{1}{2}\cos(\frac{\pi}{6}) \\
\end{array}\right).
\end{equation*}
\endgroup

One can easily check that $W_4\succeq 0$ as its eigenvalues are all nonnegative. Finally, we compute the certificate as 
\begin{equation}
    \text{Tr}(\Gamma W_4) \geq 0.
\end{equation}

Hence,

\begin{equation}
\langle A_1A_2 \rangle + \langle A_2A_3 \rangle + \langle A_3A_4 \rangle \leq \sin(\frac{\pi}{6})\langle A_1 A_4 \rangle + 3\cos(\frac{\pi}{6}).
\end{equation}

Let us now turn to the $N$-partite case. We follow the same idea, consider a moment matrix corresponding to NPA level 1 where the only expressible sets are the two-party correlators of adjacent pairs plus the pair of first and last parties, i.e., $(A_i,A_{i+1}) \forall i \in \{1,...,N-1\}$ and $(A_1,A_N)$. (Note that this is the case for the \emph{quantum cut inflation} defined by Table~\ref{tab:quantum_cut}). Concretely,

\begin{equation*}
\begin{aligned}
& \text{\hspace{10pt}$A_1$ \hspace{10pt}$A_2$ \hspace{50pt}  $\cdots$ \hspace{75pt}$A_N$} \\
 \Gamma=&\left(\begin{array}{cccccc}
1 & \langle A_1 A_2 \rangle & x & \cdots & x & \langle A_1 A_N\rangle \\
& 1 & \langle A_2 A_3\rangle & x & \cdots & x \\
& &  \ddots  & \ddots & \ddots &  \vdots \\
& & & & & \\

& & & 1 & \text{\footnotesize{ $\langle A_{N-2} A_{N-1}\rangle$}} & x \\ 
& & & & 1 & \text{\footnotesize{$\langle A_{N-1} A_{N}\rangle$}} \\
& & & & & 1
\end{array}\right),
\end{aligned}    
\end{equation*}\\
where we write all the variables as $x$ although they may take different values and use the correlators of the expressible sets. 

The elements of the witness matrix are now functions of the number of parties we consider, $N$. The concrete $N\times N$ witness matrix is of the form

\begin{widetext}
\begin{equation*}
 W_N=\left(\begin{array}{cccccc}
\frac{1}{2}\cos(\frac{\pi}{2(N-1)}) & -\frac{1}{2} & 0 & \cdots & 0 & \frac{1}{2}\sin(\frac{\pi}{2(N-1)}) \\
 & \cos(\frac{\pi}{2(N-1)}) & -\frac{1}{2} & 0 & \cdots & 0 \\
 &  & \ddots  & \ddots & \ddots & 0 \\
& & & & & \\
& & & \cos(\frac{\pi}{2(N-1)}) & -\frac{1}{2} & 0 \\ 
& & & & \cos(\frac{\pi}{2(N-1)}) & -\frac{1}{2} \\
& & & & & \frac{1}{2}\cos(\frac{\pi}{2(N-1)})
\end{array}\right)
\end{equation*}
\end{widetext}
where we only specify the upper part because it is symmetric. One can prove that it is PSD, $W_N \succeq 0$, using the Sylvester's criterion (i.e., showing that the determinant of all the principal minors are nonnegative). We provide the explicit proof in this appendix as Proposition~\ref{Sylvester}. Finally, notice that the trace of $\Gamma W$ will only involve the correlators of the expressible sets (or in other words, only the known variables of the moment matrix). Therefore, we derive an inequality that is valid for the original scenario.

\begin{equation}
\begin{split}
\sum_{i=1}^{N-1}\langle A_i A_{i+1} \rangle \leq \\ \sin(\frac{\pi}{2(N-1)})\langle A_1 A_N \rangle + (N-1)\cos(\frac{\pi}{2(N-1)}),
\end{split}
    \label{sdpN_inequality}
\end{equation}

\begin{proposition}
\label{Sylvester}
    Matrix $W_N$ is positive semidefinite for any $N\geq 3$
\end{proposition}
\begin{proof}
    We show that $W_N \succeq 0$ using the Sylvester criterion, i.e., by showing that all their leading principal minors are positive. Note that the first two principal minors are straightforward to check and that the third through the $(N-1)^{th}$ minors are tridiagonal matrices. Hence, we use the well known recurrence formula to compute the determinant of a tridiagonal matrix to show the posiivity of all the minors of size $m\leq N-1$. Such formula is 

    \begin{equation}
        f_m = a_m f_{m-1} - c_{m-1}b_{m-1}f_{m-2},
    \end{equation}
    where $f_m$ is the determinant of the minor of size $m$ and $a_m$, $b_m$, $c_m$ are the entries of the main, upper and lower diagonals respectively. The starting values for the recurrence are  $f_{-1}=0$, $f_0=1$. Note that in our case, we have that $b_m=c_m=-\frac{1}{2} \forall m \in \{1,...,N-2\}$, $a_1=\frac{1}{2}\text{cos}(\frac{\pi}{2(N-1)})$ and $a_m=\text{cos}(\frac{\pi}{2(N-1)}) \forall m \in \{2,...,N-1\}$. Substituting this in the recurrence formula:

    \begin{equation}
        f_m = \text{cos}(\frac{\pi}{2(N-1)}) f_{m-1} - \frac{1}{4}f_{m-2}  \quad  \text{for} \quad  m\geq2.
    \end{equation}

Also, the recurrence formula for the Chebysev polynomials of the first kind is

\begin{equation}
    T_m(u) = 2u T_{m-1}(u) - t_{m-2}(u)
\end{equation}
where the starting point of the recurrence is $T_0(u)=1$ and $T_1(u)=u$. Note that the Chebysev polynomials of the first kind satisfy the relation

\begin{equation}
T_m(cos(\theta))=cos(m\theta).
\end{equation}

Therefore, we can establish a relation between the recurrence formula for the determinant of a triangular matrix with the one for the Chebysev polynomials:

\begin{equation}
    f_m = 2^{-m}T_m(\text{cos}(\frac{\pi}{2(N-1)})),
\end{equation}
which is easy to verify. Then, we directly see that for any $m \in \{1,...,N-1\}$, $f_m=2^{-m}\text{cos}(m\frac{\pi}{2(N-1)})\geq 0$ with equality in the case of $m=N-1$. This shows that all leading principal minors of size less than $N$ are semipositive.

Finally, we need to check that the determinant of $W_N$ is indeed semipositive as well. To show it, we provide an eigenvector with eigenvalue 0, thus, proving that $Det(W_N)=0$. Let us first provide all the entries of the eigenvector and then check that indeed its eigenvalue is 0. The entries are given by

\begin{equation}
    v_j = \text{cos}((j-1)\frac{\pi}{2(N-1)}) \quad \text{for} \quad j=1,...,N.
\end{equation}

Now, we check the entries of the vector $W_Nv$:
\begin{itemize}
    \item First entry:
    \begin{equation*}
    \begin{split}
        (W_Nv)_1= \frac{1}{2}\text{cos}(\frac{\pi}{2(N-1)})v_1 -\frac{1}{2}v_2 + \frac{1}{2}\text{sin}(\frac{\pi}{2(N-1)})v_N \\= \frac{1}{2}(\text{cos}(\frac{\pi}{2(N-1)}) - \text{cos}(\frac{\pi}{2(N-1)}))=0
        \end{split}
        \end{equation*}

        \item Middle entries, $2\leq j\leq N-1$:
          \begin{equation*}
    \begin{split}
        (W_Nv)_j= -\frac{1}{2}v_{j-1} + \text{cos}(\frac{\pi}{2(N-1)})v_j -\frac{1}{2}v_{j+1} \\= -\frac{1}{2}(v_{j-1}+v_{j+1}) + \text{cos}(\frac{\pi}{2(N-1)})v_j=0
        \end{split}
        \end{equation*}
        where we have used that \begin{equation}
       (v_{j-1}+v_{j+1})=2\text{cos}(\frac{\pi}{2(N-1)})v_j
        \end{equation}

        \item Last entry:
        \begin{equation*}
    \begin{split}
        (W_Nv)_N \\=  \frac{1}{2}\text{sin}(\frac{\pi}{2(N-1)})v_1 -\frac{1}{2}v_{N-1} + \frac{1}{2}\text{cos}(\frac{\pi}{2(N-1)})v_N \\= \frac{1}{2}(\text{sin}(\frac{\pi}{2(N-1)}) - \text{sin}(\frac{\pi}{2(N-1)}))=0
        \end{split}
        \end{equation*}
\end{itemize}

There, as we have shown that all the leading principal minors are semipositive, we conclude that $W_N \succeq 0$.

\end{proof}

\subsection{Proof of Theorems~\ref{noiserobusttheorem4} and \nameref{noiserobusttheorem}}

Finally, we give the formal proof of Theorems~\ref{noiserobusttheorem4} and \nameref{noiserobusttheorem}. We write it only for the $N$-partite case, as they are completely analogous.

\begin{proof}    
By Lemma~\ref{lemma:multilayerG}, any inequality that we derive for the set of probability distributions compatible with $g_N^*$ is an inequality that must be satisfied by any DAG $g \in \mathcal{G}_N$. Then, we consider the cut quantum inflation defined by Table~\ref{tab:quantum_cut}. As explained in the main text, the set of the first and last parties is expressible because they are causally independent. Hence, their correlator factorize,

\begin{equation}
    \langle A_1A_N \rangle = \langle A_1 \rangle \langle A_N \rangle.
    \label{factors}
\end{equation}

Therefore, substituting Eq.~\eqref{factors} in Eq.\eqref{sdpN_inequality}, we obtain
\begin{equation*}
\sum_{i=1}^{N-1}\langle A_i A_{i+1} \rangle \leq \sin(\frac{\pi}{2(N-1)}) \langle A_1 \rangle \langle A_N \rangle + (N-1)\cos(\frac{\pi}{2(N-1)})
 \label{thm_noiseinequality_appendix}
\end{equation*}
where all the correlators correspond to injectable set (as all adjacent pairs of nodes constitute an injectable set), thus, yielding an inequality valid for the original scenario.
\end{proof}

\section{Alternative SDP inequality}
\label{appendix_alternative}

In this appendix, we provide an alternative SDP inequality that grants a slight improvement to the amount of white noise one could add to the GHZ state of four parties and still detect a violation. In this appendix, we provide the general inequality for the $N$-partite case and show that, indeed, only in the case of $N=4$ we obtain better noise values.

The SDP inequality is obtained analogously to the one explained in Appendix~\ref{SDP_appendix}. Hence, we directly provide the witness matrices, $W'_4$ and its generalization to $N$ parties.

\begin{equation*}
 W'_4=\left(\begin{array}{cccc}
\phantom{-}2 & -3 & \phantom{-}0 & \phantom{-}1 \\
-3 & \phantom{-}6 & -3 & \phantom{-}0 \\
\phantom{-}0 & -3 & \phantom{-}6 & -3 \\
\phantom{-}1 & \phantom{-}0 & -3 & \phantom{-}2 \\
\end{array}\right).
\end{equation*}

\begin{equation*}
\resizebox{\columnwidth}{!}{$
W'_N=\left(\begin{array}{cccccc}
N-2 & 1-N & 0 & \cdots & 0 & 1 \\
 & 2(N-1) & 1-N & 0 & \cdots & 0 \\
 &  & \ddots  & \ddots & \ddots & 0 \\
 &  &  & \ddots & \ddots & 0 \\
 &  &  & 2(N-1) & 1-N & 0 \\ 
 &  &  &  & 2(N-1) & 1-N \\
 &  &  &  &  & N-2
\end{array}\right)
$}
\end{equation*}

It is straightforward to check that $W'_4\succeq 0$ and, using the Sylvester criterion, it is also easy to prove that $W'_N \succeq 0$. However, as it does not lead to any improvement in the noise values for any other $N$ beyond $N=4$, we do not provide the proof of semipositivity of $W'_N$.

Therefore, the SDP inequality for the case of $N$ parties (computed as $Tr(W'_N\Gamma)\geq0$) is 

\begin{equation}
    \langle A_1 A_N \rangle \geq 1- (N-1)^2+(N-1)\sum_{i=1}^{N-1}\langle A_i A_{i+1} \rangle,
    \label{sdpN_inequality_alternative}
\end{equation}

Then, following the very same steps as in the proof of the noise robust version of Theorem \ref{GHZtheorem} but using this SDP inequality, one obtains

\begin{equation}
     \begin{split}
\left(\frac{1-\langle \tilde{A}_{\text{rest}}^1\rangle }{2}\right)^2 \left(I_{CHSH}^{\tilde{A}_{\text{rest}}^1=-1}\circ \left\{ A_1A_2\right\}\right)^2 \\ + \left( \frac{1+\langle \tilde{A}_{\text{rest}}^1 \rangle}{2}\right)^2 \left(I_{CHSH}^{\tilde{A}_{\text{rest}}^1=1}\circ \left\{ A_1A_2\right\}\right)^2 \\ + 2(1-(N-1)^2 + (N-1)I_{\text{same}})^2 
\\ \leq \left[\left( \frac{1-\langle \tilde{A}_{\text{rest}}^1 \rangle}{2}\right)^2 + \left( \frac{1+\langle \tilde{A}_{\text{rest}}^1 \rangle}{2}\right)^2\right] 8.
\label{alternativeGHZineq}
\end{split}
\end{equation}

Doing the same noise tolerance analysis with the white noise model parametrized by $v$ as in Eq.~\ref{noise_model_GHZ}, we obtain an analogous table of values for $v_{\text{min},N}$ and $f_{\text{min},N}$ given in Table~\ref{table_alternative}. Finally, to make clear that this alternative inequality only grants an improvement for $N=4$, we plot the violations of both inequalities for $N=4,\dots,10$ in Fig. \ref{fig:IneqViolationCompared}. Note that indeed only for $N=4$ the minimum value for positivity is smaller for the inequality provided in this appendix.

\begin{table}[h]
\centering
\begin{tabular}{|c|c|c|}
\hline
$N$ & $v_{\text{min},N}$ & $f_{\text{min},N}$ \\ \hline
4  & 0.9417 & 0.9454 \\ \hline
5  & 0.9617 & 0.9629 \\ \hline
6  & 0.9730 & 0.9735 \\ \hline
7  & 0.9800 & 0.9802 \\ \hline
8  & 0.9846 & 0.9847 \\ \hline
9  & 0.9878 & 0.9878 \\ \hline
10 & 0.9901 & 0.9901 \\ \hline
\end{tabular}
\caption{Numerical values (up to the fourth decimal) of $v_{\text{min},N}$ and $f_{\text{min},N}$ derived from inequality \ref{sdpN_inequality_alternative} for different values of $N$.}
\label{table_alternative}
\end{table}
\mbox{}\\

\onecolumngrid

\begin{figure*}[h]
    \centering
    \includegraphics[width=10cm]{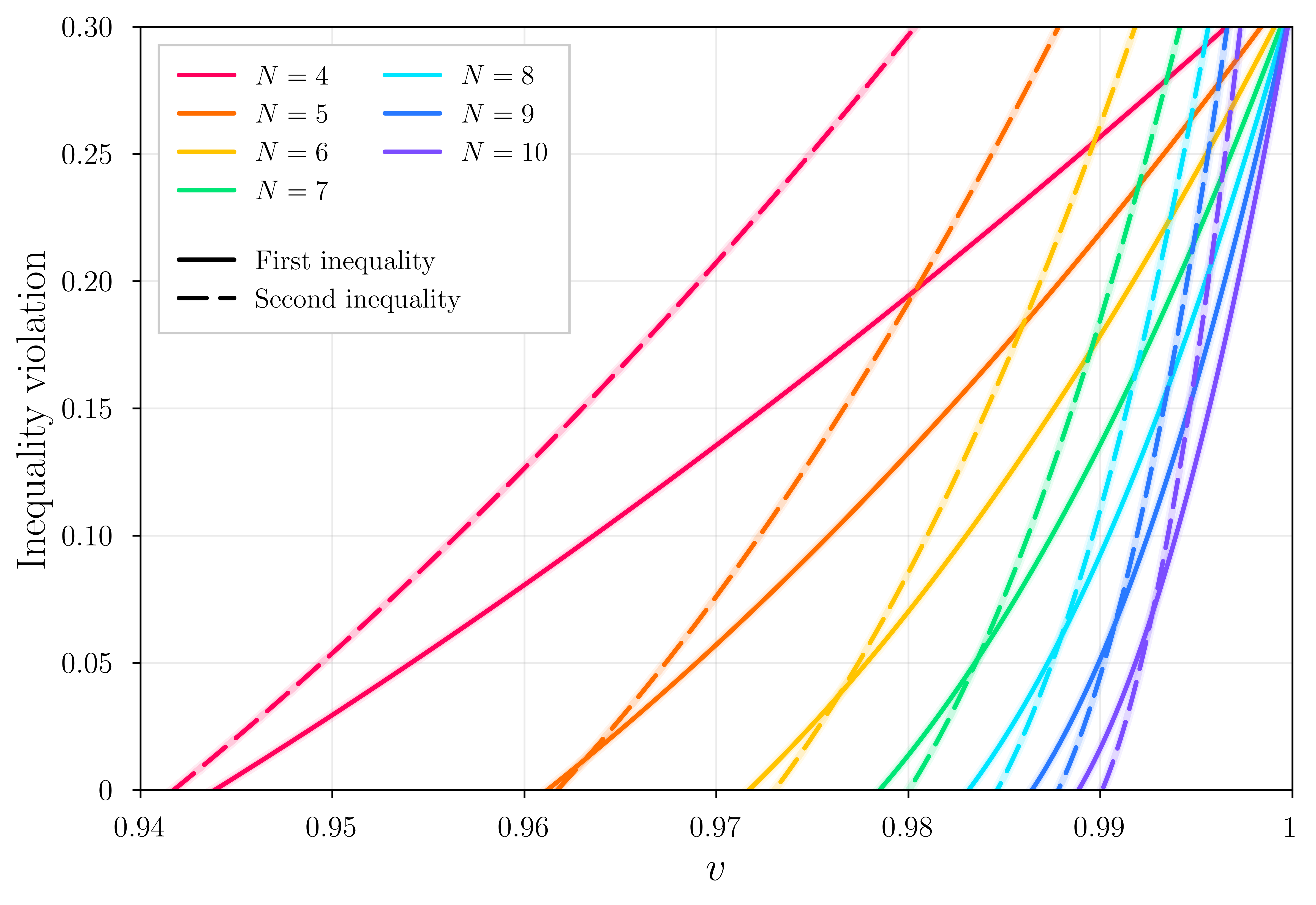}
    \caption{ Value of the inequality violation derived from inequality \ref{sdpN_inequality} (first inequality, corresponding to the curves in full lines) and inequality \ref{sdpN_inequality_alternative} (second inequality, corresponding to the curves in dashed lines), as a function of the noise parameter $v$ for different values of $N$. Note that the x-intercept correspond to to $v_{\text{min},N}$.}
    \label{fig:IneqViolationCompared}
\end{figure*}

\begin{acronym}[CGLMP]\itemsep 1\baselineskip
\acro{AGF}{average gate fidelity}
\acro{AMA}{associated measurement assemblage}

\acro{BOG}{binned outcome generation}

\acro{CGLMP}{Collins-Gisin-Linden-Massar-Popescu}
\acro{CHSH}{Clauser-Horne-Shimony-Holt}
\acro{CP}{completely positive}
\acro{CPT}{completely positive and trace preserving}
\acro{CPTP}{completely positive and trace preserving}
\acro{CS}{compressed sensing} 
\acro{DAG}{directed acyclic graph}

\acro{DFE}{direct fidelity estimation} 
\acro{DM}{dark matter}

\acro{GST}{gate set tomography}
\acro{GPT}{general probabilistic theory}
\acroplural{GPT}[GPTs]{general probabilistic theories}
\acro{GUE}{Gaussian unitary ensemble}

\acro{HOG}{heavy outcome generation}

\acro{JM}{jointly measurable}

\acro{LHS}{local hidden-state model}
\acro{LHV}{local hidden-variable model}
\acro{LOCC}{local operations and classical communication}

\acro{MBL}{many-body localization}
\acro{ML}{machine learning}
\acro{MLE}{maximum likelihood estimation}
\acro{MPO}{matrix product operator}
\acro{MPS}{matrix product state}
\acro{MUB}{mutually unbiased bases} 
\acro{MW}{micro wave}

\acro{NISQ}{noisy and intermediate scale quantum}

\acro{OPT}{operational probabilistic theory}

\acro{POVM}{positive operator valued measure}
\acro{PR}{Popescu-Rohrlich}
\acro{PVM}{projector-valued measure}

\acro{QAOA}{quantum approximate optimization algorithm}
\acro{QML}{quantum machine learning}
\acro{QMT}{measurement tomography}
\acro{QPT}{quantum process tomography}
\acro{QRT}{quantum resource theory}
\acroplural{QRT}[QRTs]{Quantum resource theories}

\acro{RDM}{reduced density matrix}

\acro{SDP}{semidefinite program}
\acro{SFE}{shadow fidelity estimation}
\acro{SIC}{symmetric, informationally complete}
\acro{SM}{Supplemental Material}
\acro{SPAM}{state preparation and measurement}

\acro{RB}{randomized benchmarking}
\acro{rf}{radio frequency}

\acro{TT}{tensor train}
\acro{TV}{total variation}

\acro{UI}{uninformative}

\acro{VQA}{variational quantum algorithm}

\acro{VQE}{variational quantum eigensolver}

\acro{WMA}{weighted measurement assemblage}

\acro{XEB}{cross-entropy benchmarking}

\end{acronym}

\end{document}